\documentclass[letterpaper]{article}
\usepackage[pdftex]{graphicx}
\graphicspath{{./figures/}{../figures/}{./}}
\usepackage{amsmath,amssymb,euscript,yfonts,psfrag,latexsym,dsfont,bbm,color,amstext,wasysym,balance,mathtools}
\usepackage{amsthm}
\newcommand{\var}{{\operatorname{var}}}
\newcommand{\Y}{{\mathcal Y}}
\newcommand{\X}{{\mathcal X}}
\newcommand{\Z}{{\mathcal Z}}
\newcommand{\JEnt}{{\rm JEnt}}
\newcommand{\Std}{{\rm Std}}
\newcommand{\SEnt}{{\rm SEnt}}

\newcommand{\by}{{\boldsymbol y}}
\newcommand{\bx}{{\boldsymbol x}}
\newcommand{\bz}{{\boldsymbol z}}

\newcommand{\bu}{{\boldsymbol u}}

\newcommand{\beps}{{\boldsymbol \epsilon}}
\newcommand{\bxi}{{\boldsymbol \xi}}

\newcommand{\mZ}{{\mathbb Z}}
\newcommand{\mR}{{\mathbb R}}
\newcommand{\cE}{{\mathcal E}}
\newcommand{\cT}{{\mathcal T}}
\newcommand{\cK}{{\mathcal K}}
\newcommand{\F}{{\mathcal F}}

\newcommand{\f}{{\operatorname{f}}}

\newtheorem{prop}{Proposition}
\newtheorem*{remark}{Remark}
\usepackage{url}

\begin{document}

\title{Minimum-entropy causal inference and its application in brain network analysis}

\date{}
\author{Lipeng Ning
\thanks{L. Ning is with the Psychiatry Neuroimaging Laboratory, Department
of Psychiatry, Brigham and Women's Hospital, Harvard Medical School, Boston, MA 02215 USA,
e-mail: lning@bwh.harvard.edu.}}

\markboth{Journal of \LaTeX\ Class Files,~Vol.~xx, No.~xx, October~2021}%
{L. Ning: ME causal inference and its application in brain network analysis}


\maketitle

\begin{abstract}
Identification of the causal relationship between multivariate time series is a ubiquitous problem in data science. 
Granger causality measure (GCM) and conditional Granger causality measure (cGCM) are widely used statistical methods for causal inference and effective connectivity analysis in neuroimaging research. 
Both GCM and cGCM have frequency-domain formulations that are developed based on a heuristic algorithm for matrix decompositions. 
The goal of this work is to generalize GCM and cGCM measures and their frequency-domain formulations by using a theoretic framework for minimum entropy (ME) estimation. The proposed ME-estimation method extends the classical theory of minimum mean squared error (MMSE) estimation for stochastic processes. It provides three formulations of cGCM that include Geweke's original time-domain cGCM as a special case. But all three frequency-domain formulations of cGCM are different from previous methods.
Experimental results based on simulations have shown that one of the proposed frequency-domain cGCM has enhanced sensitivity and specificity in detecting network connections compared to other methods. 
In an example based on \emph{in vivo} functional magnetic resonance imaging, the proposed frequency-domain measure cGCM can significantly enhance the consistency between the structural and effective connectivity of human brain networks.
\end{abstract}



\section{Introduction}\label{sec1}
Causal inference between multi-variate time series is a fundamental problem in data science. It is particularly relevant in neuroscience research where time-series neuroimaging data, such as functional magnetic resonance imaging (fMRI), electroencephalogram (EEG), and magnetoencephalography (MEG), are used to investigate interactions between brain regions that form the so-called brain networks.
Brain cognition, function, and memory all involve dynamic and most likely asymmetric interactions between multiple brain regions \cite{Luo2013a,Sporns2013,Park2013,Tanimizu2017,Martinez2018}.
Identification of the direction of information flow in brain networks is critical to understanding the mechanism of brain functions and developing markers for brain diseases \cite{Stam2007,Fitzsimmons2013,Fornito2015}. 

The influence that a brain node exerts over another under a network model is usually referred to as effective connectivity (EC) \cite{Harrison2003}. 
Several methods have been developed for EC analysis of brain networks, such as dynamic causal modeling (DCM) \cite{Stephan2008,Friston2013}, and the Granger causality measure (GCM) and conditional GCM (cGCM) \cite{Granger1969,Geweke1982,Geweke1984,Chen2006,Solo2008,Barnett2009,Deshpande2009,Marinazzo2011,Deshpande2012,Seth2015,Wang2020}, and directed transfer function (DTF) \cite{Kaminski1991,Kaminski2001}.
DCM uses a deterministic multiple-input and multiple-output system with time-varying coefficients to model the dynamic coupling of the underlying variables. GCM and cGCM characterize the stochastic dependence between time series by quantifying the importance of past values of one variable to the prediction of another one.
 DTF is derived based on a similar model as GCM and cGCM to quantify the dependence of multivariate time series in the spectral domain. More detailed comparisons and explanations of these methods can be found in \cite{Kaminski2001,Friston2013,Friston2014}.

The goal of this work is to introduce an information theoretical framework for causal inference that extends existing GCM and cGCM methods. In this context, several related information-theoretic interpretations or generalizations have already been developed for GCM or cGCM including transfer entropy (TE) and directed information (DI)\cite{1057325,Solo2008,Barnett2009,Sciences2011,Vicente2011,Quinn2011,Amblard2011a,Amblard2011}.  
\cite{1057325} first pointed out the relation between entropy differences and Geweke's GCM by introducing a causality measure as the difference between the code length for encoding one variable using its own past values and the code length based on the joint past values with other variables. 
\cite{HLAVACKOVASCHINDLER20071} introduced several methods for causal inference based on mutual information and TE. \cite{HLAVACKOVASCHINDLER20071} and \cite{Solo2008} showed that the GCM between a pair of time series are also mutual information measures. \cite{Barnett2009} showed the equivalence between GCM and TE for Gaussian processes.
The equivalence between GCM and TE has also been extended to generalized Gaussian probability distributions \cite{Sciences2011}, model-free tools for effective connectivity analysis in neuroscience \cite{Vicente2011} and nonlinear causality detections \cite{doi:10.1098/rsos.200863}. 
In \cite{Quinn2011} introduced the DI measure to infer causal relationships in neural spike train recordings based on previous works in \cite{1057325}. DI is derived based on a modified mutual information measure and can be decomposed as the sum of the TE and a term that quantifies the instantaneous coupling \cite{Amblard2011a}.
More detailed overviews about GCM and related information-theoretic formulations can be found in  \cite{Seghouane2012,Amblard2013}.

The causality measures introduced in this work are developed based on a method of minimum entropy (ME) estimation which is different from the TE- or DI-based methods. 
The problem of ME estimation is similar to the classical minimum mean square error (MMSE) estimation where a linear dynamic filter is applied to one time series to predict another one based on a joint model to minimize the entropy rate of the residual process. 
The ME method not only provides an alternative information-theoretic interpretation of time-domain GCM or cGCM but also introduces the ME processes whose power spectral density (PSD) functions are used to derive frequency-domain measures. 
The ME-based method was first investigated in our previous work \cite{Ning2018} to generalize partial coherence analysis and frequency-domain GCM. 
This work further extends this framework to introduce three types of ME-based causality measures to quantify the conditional causality between multivariate time series. 
The relationship between the three methods analyzed and their performances are compared using both simulations and \emph{in vivo} resting-state fMRI (rsfMRI) data of human brains.

The organization of this paper is as below. Section \ref{sec:background} introduces preliminary knowledge on GCM, cGCM, and related information-theoretic formations such as TE, and DI.
Section \ref{sec:ME_cGCM} introduces the ME-based formulation of GCM, three ME-based formulations of conditional causality measures, and the relationship among these measures.
Section \ref{sec:experiment} introduces two experiments based on simulations and \emph{in vivo} rsfMRI data to demonstrate the performance of the proposed methods.
Section \ref{sec:results} presents the experimental results.
The discussion and conclusions are presented in Section \ref{sec:discussion} and Section \ref{sec:conclusions}, respectively.
The Appendix introduces the details of the computational algorithms.

\section{On GCM, cGCM, transfer entropy and directed information}\label{sec:background}

\subsection{Notations and background on power spectral analysis}
Let $\bu_t=(\bx_t;\by_t;\bz_t)\in \mR^n$ denote a zero-mean wide-sense stationary Gaussian process and $\bx_t\in  \mR^{n_x}, \by_t\in  \mR^{n_y}, \bz_t\in  \mR^{n_z}$, $n_x+n_y+n_z=n$.
For the sub-process $\bx_t$, let 
\begin{align}
    \Sigma_{xx} = \var(\bx \mid \X_{-})
\end{align}
denote the one-step ahead prediction error variance of $\bx_{t}$ given its past values $\X_{-}=\{\bx_{t-k}, k \in \mZ^+ \}$, where $ \mZ^+$ denotes the set of positive integers and the subscript $t$ is omitted for simplicity. 
Moreover, let $\X = \{\bx_t, \X_-\}$ denote the set of past and current values of $\bx_t$. Similarly, the $\Y_- (\Y)$ and $\Z_-(\Z)$ are used to denote the set of past (past and current) values of $\by_t$ and $\bz_t$, respectively.

The prediction error $\bx \mid \X_{-}$ follows a zero-mean Gaussian distribution and its differential entropy is equal to \cite{Papoulis1985}
\begin{align}\label{eq:h}
h(\bx \mid \X_{-}) = \frac12 n_x(1+\ln(2\pi)) + \frac12 \ln \det \Sigma_{xx}.
\end{align}
An alternative interpretation of $h(\bx_{t} \mid \X_{-})$ is provided by the Shannon entropy rate of the process $\bx_t$ which is defined as
\begin{align}\label{eq:entropy_rate}
h(\bx):= \lim_{m \rightarrow \infty} \frac{1}{m}h([\bx_{t}; \bx_{t-1}; \ldots; \bx_{t-m}]),
\end{align}
and satisfies that \cite{Kolmogorov1958,Mulherkar2018}
\begin{align}\label{eq:entropy_equal}
h(\bx)= h(\bx \mid \X_{-}).
\end{align}

Let $S_{\bx}(\theta)$ denote the PSD function of $\bx_{t}$ with $\theta\in [-\pi, \pi]$.
If $S_{\bx}(\theta)$ is positive definite for all $\theta\in [-\pi, \pi]$, then the Kolmogorov-Szeg\"{o} equation \cite{Wiener1957} indicates that 
\begin{align}\label{eq:OmegaPSD}
    \det( \Sigma_{xx} ) = \exp\left(\frac1{2\pi} \int_{-\pi}^{\pi} \ln \det S_{\bx}(\theta) d\theta\right),
\end{align}
which implies that $\det(\Sigma_{xx})$ is equal to the geometric mean of $\det S_{\bx}(\theta)$.

\subsection{Vector autoregressive representations}
For the joint process $\bu_t$, let $
G(L) \bu_t = \beps_t$,
denote the vector autoregressive (VAR) presentations where $\beps_t$ represents the optimal prediction error of $\bu_t$ based on its past values, $L$ denotes the lag operator such that $L^k \bu_t= \bu_{t-k}$ and $G(L) = I_{n} +\sum _{k=1}^\infty G_k L^k$.
Based on the decomposition of $\bu$, let the above VAR model be decomposed as
\begin{align}\label{eq:G_abc}
\left[  \begin{matrix}G_{xx}(L)& G_{xy}(L)& G_{xz}(L)\\  G_{yz}(L)& G_{yy}(L)& G_{yz}(L)\\  G_{zx}(L)& G_{zy}(L)& G_{zz}(L)\end{matrix} \right]\left[\begin{matrix} \bx\\ \by\\ \bz\end{matrix} \right]= \left[\begin{matrix} \beps_{x}\\ \beps_{y}\\ \beps_{z}\end{matrix} \right],
\end{align}
where the time index $t$ is omitted for notational simplicity and the covariance matrix of $\beps$ is denoted by
\begin{align}\label{eq:Omega_eps}
\cE(\beps \beps^T)=\Omega = \left[  \begin{matrix}\Omega_{xx}& \Omega_{xy}& \Omega_{xz}\\  \Omega_{yz}& \Omega_{yy}& \Omega_{yz}\\  \Omega_{zx}& \Omega_{zy}& \Omega_{zz}\end{matrix} \right] .
\end{align}

Assume that the subprocess $(\bx;\by)$ extracted from $\bu_t$ has the following VAR representation
\begin{align}\label{eq:VARy12}
    \left[\begin{matrix} \hat G_{xx}(L)& \hat G_{xy}(L) \\  \hat G_{yx}(L)& \hat G_{yy}(L) \end{matrix}\right] \left[\begin{matrix}\bx\\ \by \end{matrix} \right] = \left[\begin{matrix} \hat \beps_{x} \\ \hat \beps_{y}\end{matrix}\right],
\end{align}
where the covariance matrix of $(\hat \beps_{x}; \hat \beps_{y})$ is denoted by
\begin{align}\label{eq:Omega12}
\cE\left(\left[\begin{matrix} \hat \beps_{x} \\  \hat \beps_{y} \end{matrix}\right]
\left[\begin{matrix} \hat \beps_{x}^T &  \hat \beps_{y}^T \end{matrix}\right]
\right)
= \left[\begin{matrix} \hat \Omega_{xx} & \hat \Omega_{xy} \\ \hat \Omega_{yx} & \hat \Omega_{yy} \end{matrix} \right].
\end{align}
The symbol $\hat{.}$ distinguishes the notations in \eqref{eq:VARy12} and \eqref{eq:Omega12} from those in \eqref{eq:G_abc} and \eqref{eq:Omega_eps}. 
It should be noted that the VAR representation for the subprocess $(\bx;\by)$ should be consistent with the model of the joint process \cite{Barnett2015,10.1162/NECO_a_00828}. The state-state representation for the joint process $\bu$ and spectral factorization algorithms provide an approach to deriving consistent models of sub-processes which is explained in the Appendix.

{
\subsection{GCM, cGCM and the frequency-domain formations}

Based on the definitions, it clearly holds that $\det(\Sigma_{xx})\geq \det(\hat \Omega_{xx})$ since using the past values of both $\bx$ and $\by$ improves the prediction accuracy for $\bx$ compared to using only the past values of $\bx$. Similarly, $\det(\hat \Omega_{xx})\geq \det(\Omega_{xx})$ since adding the past values of $\bz$ further improves the prediction accuracy. The GCM and cGCM measures are derived based on the above inequalities to quantify the improvement in prediction accuracy by adding more variables. Specifically, the GCM from $\by$ to $\bx$ is defined as \cite{Geweke1982}
\begin{align}\label{eq:GCMstd}
\F_{\by\rightarrow \bx} &= \ln \frac{\det \Sigma_{xx}}{\det \hat \Omega_{xx}}.
\end{align}
The cGCM from $\bz$ to $\bx$ conditional on the past values of $\by$ is defined as \cite{Geweke1984}
\begin{align}\label{eq:FOrig}
\F^{\rm Std}_{\bz \rightarrow \bx | \by} = \ln \frac{\det \hat \Omega_{xx}}{\det \Omega_{xx}},
\end{align}
where the superscript $^{\rm Std}$ is used to distinguish the standard method in  \cite{Geweke1984} from the results introduced in the next section. 
It is noted that the cGCM $\F^{\rm Std}_{\by \rightarrow \bx | \bz}$ can be derived based on a joint model for $(\bx; \bz)$ similar to \eqref{eq:VARy12}.

Both GCM and cGCM have frequency-domain formulations introduced by Geweke in \cite{Geweke1982,Geweke1984}.
The frequency-domain GCM (fGCM) is defined as
\begin{align}\label{eq:Geweke_fGCM}
\f_{\by\rightarrow \bx}^{\rm Geweke}(\theta) &= \ln \frac{\det S_{\bx}(\theta)}{\det (S_{\bx}(\theta)- \hat H_{xy}(e^{i\theta}) \hat \Omega_{y|x} {\hat H_{xy}(e^{i\theta})}^* )},
\end{align}
where $\hat \Omega_{y|x}=\hat \Omega_{yy}-\hat \Omega_{yx}\hat \Omega_{xx}^{-1}\hat \Omega_{xy}$ and 
\begin{align}\label{eq:hatHL}
\hat H(L) =\left[\begin{matrix} \hat H_{xx}(L)& \hat H_{xy}(L) \\  \hat H_{yx}(L)& \hat H_{yy}(L) \end{matrix}\right] =\left[\begin{matrix} \hat G_{xx}(L)& \hat G_{xy}(L) \\  \hat G_{yx}(L)& \hat G_{yy}(L) \end{matrix}\right] ^{-1}.
\end{align}
The frequency-domain cGCM (fcGCM) is derived based on a similar matrix-decomposition algorithm which is provided in \eqref{eq:Geweke_fcGCM}.

It is noted that the denominator of $\f^{\rm Geweke}_{\by\rightarrow \bx}(\theta)$ is not the PSD function of a time series. 
The fGCM is defined based on a heuristic algorithm to satisfy that
\begin{align}\label{eq:fGCM_mean}
\frac{1}{2\pi} \int_{-\pi}^{\pi}    \f_{\by\rightarrow \bx}^{\rm Geweke}(\theta) d\theta = \F_{\by\rightarrow \bx}.
\end{align}
Several studies have shown that results based on fGCM or fcGCM are not consistent with physiological knowledge of brain activities \cite{Faes2017,Stokes2017}.
To overcome the limitations, \cite{Hosoya2017} has introduced a modified causality measure whose frequency-domain formulation is expressed in terms of PSD functions of two-time series derived from the joint process. The modified time-domain causality measure is different from the GCM and is always no larger than GCM. 
The ME estimation approach introduced in our previous work \cite{Ning2018} has introduced an alternative frequency-domain formulation that is not only expressed in terms of PSD functions of time series but also provides the same mean value as the original fGCM given in \eqref{eq:fGCM_mean}. More details of the ME-based fGCM are provided in the next section.

\subsection{On transfer entropy and directed information}
Transfer entropy (TE) was introduced in \cite{Schreiber2000} to measure directed dependence between discrete-valued time series and was extended in \cite{Kaiser2002} for continuous process. The TE from $\by$ to $\bx$ is defined as
\begin{align}\label{eq:TE}
\cT_{\by\rightarrow \bx} = h(\bx | \X_{-}) - h(\bx | \X_-, \Y_-),
\end{align}
where $\bx|(\X_-, \Y_-)=\hat \beps_x$ given by \eqref{eq:VARy12}.
Based on \eqref{eq:h}, \eqref{eq:Omega12} and \eqref{eq:GCMstd}, it can be shown that
\begin{align}\label{eq:FequalT}
\F_{\by\rightarrow \bx}= 2 \cT_{\by\rightarrow \bx}, 
\end{align}
which draws the relationship between TE and GCM for Gaussian processes. The above equation is extended in \cite{Barnett2009} for conditional causality where it was shown that
\[
\F_{\by\rightarrow \bx \mid \bz}= 2 \cT_{\by\rightarrow \bx\mid \bz} . 
\]
The TE measure is closely related to the directed information (DI) for causal inference  \cite{Massey1990,Quinn2013,Amblard2011a,Amblard2011}. Consider two segments of the time series $\X^n=(\bx_1; \ldots;\bx_n)$ and $\Y^n=(\by_1; \ldots;\by_n)$. Then the DI is defined as
\begin{align}
I(\Y^n \rightarrow \X^n) &:= \sum_{i=1}^n I(\Y^i; \bx_i | \X^{i-1}),\\
&= h(\X^n) - \sum_{i=1}^n h(\bx_i \mid \X^{i-1}, \Y^i),
\end{align}
where $I(\cdot; \cdot | \cdot)$ denotes the conditional mutual information.
Then the limit of the rate of DI is given by
 \begin{align}
I_{\infty}(\by\rightarrow \bx) &:= \lim_{n\rightarrow \infty} \frac1n I(\Y^n \rightarrow \X^n), \\
&=\lim_{n\rightarrow \infty} \frac1n( h(\X^n) - \sum_{i=1}^n h(\bx_i \mid \X^{i-1}, \Y^i)),\\
&= h(\bx) - h(\bx| \X_-,\Y),\\
&=h(\bx|\X_-) - h(\bx| \X_-,\Y)\label{eq:DI}
\end{align}
where the last equation is obtained based on \eqref{eq:entropy_equal}.
The difference between \eqref{eq:DI} and \eqref{eq:TE} lies in that $\Y$ contains the instantaneous measurement that is not included in $\Y_-$. Thus the rate of DI can be decomposed as TE and another term that quantifies the instantaneous coupling as shown in \cite{Amblard2011a,Amblard2011}. A generalization of DI for conditional causality analysis is discussed in \cite{Amblard2011a,Amblard2011}.

It is noted that \eqref{eq:DI} also coincides with the GCM with instantaneous feedback introduced in \cite{Geweke1982}. Based on the results in \cite{Geweke1982} and \eqref{eq:VARy12}, it can be shown that the covariance matrix of $\bx|(\X_-,\Y)$ is equal to
\[
\hat\Omega_{x|y}=\hat \Omega_{xx}-\hat \Omega_{xy}\hat \Omega_{yy}^{-1}\hat \Omega_{yx},
\]
and
\[
I_{\infty}(\by\rightarrow \bx) = \frac12 \ln \frac{\det \Sigma_{xx}}{\det(\hat\Omega_{x|y})}.
\]
In the following section, causal inference is only derived based on the strict past values as considered in the definitions of TE, GCM, and cGCM.  

}

\section{Minimum-entropy (ME) causal inference}\label{sec:ME_cGCM}

\subsection{ME-based GCM and fGCM}
The ME estimator of $\bx$ given the past values of $\by$ is defined as the filtered process $F(L) \by$ such that the entropy of the residual process 
\begin{align}\label{eq:y21t}
    \bx||\by:= \bx - F(L) \by
\end{align}
is minimized \cite{Ning2018}.
The problem in Eq. \eqref{eq:y21t} is closely related to the classical MMSE estimation method. 
Consider two zero-mean Gaussian random variables $X$ and $Y$. The MMSE estimator of $X$ given $Y$ is equal to the conditional expectation $\cE(X|Y=y)=F y$ where $F=\cE(XY^T)\cE(YY^T)^{-1} y$. In this case, the predictor $Fy$ not only minimizes the mean square error by also the entropy of the prediction error. 
For time-series data, the MMSE estimator has been extended to the Wiener filter and the Kalman filter \cite{Wiener1956,Kalman1960} to predict stochastic processes using dynamic filters where the constant predictor $F$ is generalized to a dynamic filter.

Though the problem of ME estimation is closely related to MMSE estimation, the optimal solution to Eq. \eqref{eq:y21t} is different from the Wiener or the Kalman filter. The difference lies in the fact that the mean-squared value is equal to the arithmetic mean of the trace of the PSD function of the prediction error process whereas the entropy is related to the geometric mean of the determinant of the PSD function as in Eq. \eqref{eq:OmegaPSD}. 

For the VAR model in \eqref{eq:VARy12}, if $\hat G_{xx}(L)$ has a stable inversion, then 
the ME residual is given by
\begin{align}\label{eq:y12}
   \bx||\by = \hat G_{xx}(L)^{-1} \hat \beps_{x}, 
\end{align}
with the ME filter given by \cite{Ning2018} 
\begin{align}\label{eq:ME_filter}
F(L)=-\hat G_{xx}(L)^{-1}\hat G_{xy}(L).
\end{align}
It was pointed out in \cite{Ning2018} that the entropy rate of the residual process is equal to
\[
 h(\bx||\by) =  h(\hat \beps_{x})= h(\bx| \X_-,\Y_-).
\]
Therefore, the GCM satisfies that
\begin{align}\label{eq:GCM}
    \F_{\by \rightarrow \bx} = 2\left( h(\bx) - h(\bx||\by) \right),
\end{align}
which quantifies the extent to which the past values of $\by$ can reduce the entropy rate of $\bx$.

Although \eqref{eq:GCM} is similar to the TE-based formulation in \eqref{eq:FequalT}, a major difference between the two methods is that the ME method has introduced a new time series $\bx||\by$ which can be calculated based on the measured data. Based on $\bx||\by$, a new fGCM method was introduced in \cite{Ning2018} as below
\begin{align}\label{eq:fGCM}
    \f_{\by \rightarrow \bx}^{\rm Ent}(\theta) = \ln \frac{\det S_{\bx}(\theta)}{\det S_{\bx||\by}(\theta)}.
\end{align}
Similar to the original fGCM $\f_{\by\rightarrow \bx}^{\rm Geweke}(\theta)$, $\f_{\by \rightarrow \bx}^{\rm Ent}(\theta)$ also satisfies that
\begin{align}
    \frac{1}{2\pi} \int_{-\pi}^{\pi} \f_{\by \rightarrow \bx}^{\rm Ent}(\theta)d \theta=\F_{\by \rightarrow \bx}.
\end{align}
It is noted again that the definition of $\f_{\by\rightarrow \bx}^{\rm Geweke}(\theta)$ does not explicitly provide a new time series.

For multivariate time series, the ME-based causality from $\by$ to $\bx$ conditional on the past value of $\bz$ can have three different formulations as illustrated in Fig. \ref{fig:algorithm}. Fig. (1a) illustrates the method of ME estimation and the corresponding GCM and fGCM measures. Figs. (1c) to (1e) illustrate three methods to remove the impact of $\by$ and/or $\bz$ on $\bx$. These procedures are equivalent for MMSE estimation problems but they provide different solutions for ME-based conditional causal inference which are derived in the following subsections.

\begin{figure*}
\centering
\includegraphics[width=.65\linewidth]{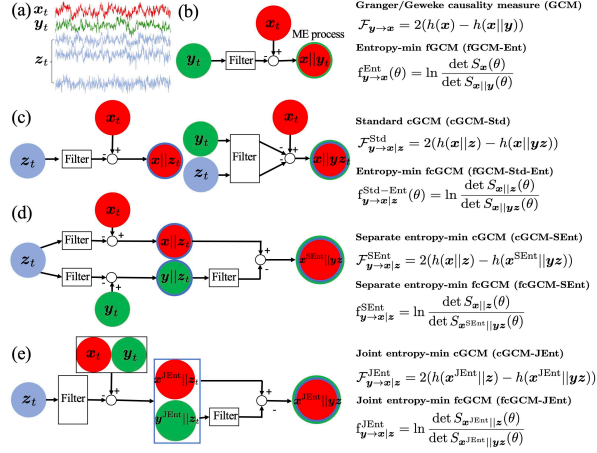}
\caption{Illustration of the computation algorithms for ME-based causality and conditional causality measures. (a) illustrates three sets of time series. (b) illustrates the method for ME-based GCM from $\by$ to $\bx$. (c) illustrates the ME-based formulation of the standard cGCM and fcGCM. (d) demonstrates the algorithm for separate entropy minimization based cGCM and fcGCM. (e) shows the algorithm for joint entropy minimization based cGCM and fcGCM.}
\label{fig:algorithm}
\end{figure*}

\subsection{ME-based formulation of standard cGCM}
The standard cGCM introduced in \cite{Geweke1984} can be formulated based on the ME estimators using the algorithm illustrated in Fig. \ref{fig:algorithm}(c).
To derive this formulation, 
let $\bx||\bz$ and $\bx||\by\bz$ denote the ME prediction error of $\bx$ by using the past of $\bz$ and the joint process of $(\by; \bz)$, respectively. 
Then, the standard cGCM is equal to
\begin{align}\label{eq:cGCMstd}
    \F^{\rm Std}_{\by\rightarrow \bx \mid \bz} = 2(h(\bx||\bz) - h(\bx||\by\bz)),
\end{align}
where $\bx||\by\bz$ represents the ME residual obtained by regressing out the past values of the joint process $(\by; \bz)$ in $\bx$ whose solution is given in Eq. \eqref{eq:by_123t_G}.
Based on Eq. \eqref{eq:fGCM}, the following function 
\begin{align}\label{eq:fcGCMstd}
    \f^{\rm Std-Ent}_{\by\rightarrow \bx \mid \bz}(\theta) =\ln \frac{\det S_{\bx|| \bz}(\theta)}{\det S_{\bx|| \by\bz}(\theta)}
\end{align}
provides a frequency-domain measure of conditional causality between the time series.
The solution for $\bx||\bz$ can be derived following Eq. \eqref{eq:y12} based on the joint model of $(\bx; \bz)$.
Similar to the standard fcGCM, $\f^{\rm Std-Ent}_{\by\rightarrow \bx \mid \bz}(\theta)$ also satisfies that
\begin{align}
    \frac{1}{2\pi} \int_{-\pi}^{\pi} \f^{\rm Std-Ent}_{\by\rightarrow \bx \mid \bz}(\theta)  d\theta = \F^{\rm Std}_{\by\rightarrow \bx \mid \bz}.
\end{align}
It is note that the model from the joint process $(\bx; \bz)$ should be consistent with the joint model for $(\bx; \by; \bz)$ and any other subprocess. Such a solution can be obtained based on the state-space representation and spectral factorization algorithms which are explained in the Appendix.

Though Eq. \eqref{eq:cGCMstd} and Eq. \eqref{eq:TE} have similar expressions, the underlying rationales between the two methods are different. First, Eq. \eqref{eq:cGCMstd} quantifies the difference between the entropy rate of two time-series $\bx||\by$ and $\bx||\by\bz$ which are not necessarily white Gaussian processes. On the other hand, Eq. \eqref{eq:TE} quantifies the difference between conditional entropy of two Gaussian random variables $\bx| (\X_{-},\Z_{-})$ and $\bx|(\X_{-}, \Y_{-},\Z_{-})$ which can be considered as elements of two white innovation processes.  
Second, the ME-based spectral formulation in Eq. \eqref{eq:fcGCMstd} compares the PSD functions of two time-series introduced in Eq. \eqref{eq:cGCMstd}. But the original spectral formulations of GCM \cite{Geweke1982} or cGCM in \cite{Geweke1984} are developed based on a heuristic matrix-decomposition method to satisfy Geweke's requirements that the spectral measures are nonnegative and are equal to the time-domain measure on average \cite{Geweke1982,Geweke1984}, see Eq. \eqref{eq:Geweke_fGCM} and Eq. \eqref{eq:Geweke_fcGCM} for more details. 
But the decomposed PSD functions are not related to any time series. For this reason, it was argued in \cite{Chicharro2011} that the TE formulation lacks a spectral representation in terms of processes based on measured data. On the other hand, the ME-based frequency-domain measures in Eq. \eqref{eq:fGCM} and Eq. \eqref{eq:fcGCMstd} may take negative values at certain frequencies, though their mean values are non-negative. 
Negative frequency-domain measures can occur when the ME filter enhances the power at certain frequencies and reduce the power at other potentially more important frequencies to ensure the overall entropy is minimized.

\subsection{Separate ME estimation-based cGCM}
{For three jointly Gaussian random variables $X, Y$, and $Z$, the conditional expectation of $X$ given $Y$ and $Z$ have the following equivalent formulations 
\begin{align}\label{eq:CondExpec}
\cE_x(X \mid Y=y, Z=z) =  \cE_x(  (X| Z=z) \mid (Y|Z=z)=y).
\end{align}
It indicates that the MMSE estimation of $X$ using the joint random variable $(Y; Z)$ is the same as the result from a two-step procedure that takes the expectation on $Z$ and $Y$ separately.}
The ME-based estimator for Gaussian processes can be considered as a generalization of the conditional expectation or the MMSE estimator. Thus, an alternative formulation for cGCM can be derived analogous to the two-step method in Eq. \eqref{eq:CondExpec}, as illustrated in Fig.  \ref{fig:algorithm}(d), which is introduced below.

First, two causal filters are separately applied to obtain the ME estimation of $\bx$ and $\by$ based on the past measurements of $\bz$ with the prediction error processes denoted by $\bx||\bz$ and $\by||\bz$, respectively.
Next, another filter is applied to $\by||\bz$ to predict $\bx||\bz$ with the ME residual process denoted by 
\begin{align}\label{eq:hatxyz}
 \bx^{\rm SEnt}||\by\bz : =(\bx|| \bz) || (\by||\bz).
\end{align}
Then, the corresponding cGCM is defined by
\begin{align}
    \F^{\SEnt}_{\by\rightarrow \bx | \bz} &:=  \F_{\by||\bz \rightarrow \bx|| \bz}\\
    &=2(h(\bx|| \bz)-h(\bx^{\rm SEnt}||\by\bz)),\label{eq:cGCMSEnt}
\end{align}
where $\SEnt$ represents ``separate entropy" minimization.
The corresponding fcGCM is given by
\begin{align}
    \f^{\SEnt}_{\by\rightarrow \bx | \bz}(\theta) &:= \f_{\by||\bz \rightarrow \bx|| \bz}^{\rm Ent}(\theta)\\
    &=\ln \frac{\det S_{\bx||\bz}(\theta)}{\det S_{\bx^{\rm SEnt} ||\by\bz}(\theta)}.
\end{align}
More details about the computational algorithms for $F^{\SEnt}_{\by\rightarrow \bx | \bz}$ and $\f^{\SEnt}_{\by\rightarrow \bx | \bz}(\theta)$ based on state-space representations are provided in Appendix.

\subsection{cGCM based on joint ME estimation}

{Consider three random variables $X, Y, Z$, below are two equivalent formulations for conditional expectations
\begin{align}
\left(\begin{matrix} \cE(X | Z=z );\;  \cE(Y | Z=z )\end{matrix} \right) = \cE\left( \left(\begin{matrix}X ;\; Y \end{matrix}\right) \mid Z=z\right),
\end{align}
which trivially implies that the MMSE estimation of joint random variable $(X; Y)$ given $Z$ is equal to the joint of the MMSE estimation of each random variable.
But the entropy of joint stochastic processes is not equal to the sum of the entropy of individual processes because of the coupling between the processes.
By exploring the difference between the two formulations, the third type of cGCM can be derived following the approach illustrated in Fig. \ref{fig:algorithm}(e). }

First, a filter is applied to $\bz$ to predict the joint process $\left( \bx; \by\right)$ to minimize the entropy of the residual process denoted by 
$\left( \bx^{\rm JEnt}||\bz;  \by^{\rm JEnt} ||\bz \right)$, where $\JEnt$ represents ``joint entropy" minimization.
{The main difference between $\bx^{\rm JEnt}||\bz$ and $\bx||\bz$ is that $\bx||\bz$ does not consider $\by$ when regressing out the past values of $\bz$ from $\bx$.}
Next, a filter is applied to $\by^{\rm JEnt}||\bz$ to predict the ME estimation of $\bx^{\rm JEnt}||\bz$ with the residual process being denoted by 
\begin{align}\label{eq:tildexyz}
 \bx^{\rm JEnt}||\by\bz := (\bx^{\rm JEnt} ||\bz)|| (\by^{\rm JEnt} ||\bz).
\end{align}
Then, the third type of cGCM is defined by 
\begin{align}
    \F^{\JEnt}_{\by\rightarrow \bx \mid \bz} &:= \F_{\by^{\rm JEnt} ||\bz  \rightarrow \bx^{\rm JEnt} ||\bz }\\
    &=2(h(\bx^{\rm JEnt}||\bz)-h(\bx^{\rm JEnt}||\by\bz)). \label{eq:FJEnt}
\end{align} 
The corresponding frequency-domain cGCM is defined by
\begin{align}
    \f^{\JEnt}_{\by\rightarrow \bx \mid \bz}(\theta) &:= \f_{\by^{\rm JEnt} ||\bz  \rightarrow \bx^{\rm JEnt} ||\bz }(\theta)\\
    &=\ln \frac{\det S_{\bx^{\rm JEnt} ||\bz}(\theta)}{\det S_{\bx^{\rm JEnt} ||\by\bz}(\theta)}),\label{eq:cGCMJEnt}
\end{align}
whose mean value is equal to $\F^{\JEnt}_{\by\rightarrow \bx \mid \bz}$. 

\subsection{On the relationship between ME-based cGCM methods}

Though the three ME-based cGCM methods are motivated based on equivalent formulations of MMSE estimation, their values are different and satisfy the following Proposition.
\begin{prop}\label{prop1}
For a joint zero-mean wide-sense stationary Gaussian processes $(\bx; \by; \bz)$, the $ \F^{\Std}_{\by \rightarrow \bx \mid \bz}, \F^{\SEnt}_{\by \rightarrow \bx \mid \bz},\F^{\JEnt}_{\by \rightarrow \bx \mid \bz}$ measures defined in \eqref{eq:cGCMstd}, \eqref{eq:cGCMSEnt}, and \eqref{eq:FJEnt}, respectively, satisfy that
\begin{align}\label{eq:ineq}
\F^{\JEnt}_{\by \rightarrow \bx \mid \bz} \geq  \F^{\Std}_{\by \rightarrow \bx \mid \bz}\geq \F^{\SEnt}_{\by \rightarrow \bx \mid \bz}\geq 0.
\end{align}
\end{prop}
\begin{proof}
Based on the definitions of $\bx^{\rm JEnt} ||\by\bz$ and $\bx||\by\bz$, it can be shown that
\begin{align}\label{eq:ineq1a}
 \bx^{\rm JEnt} ||\by\bz = \bx||\by\bz = G_{xx}^{-1}(L) \beps_x,
\end{align}
where more details can be found in Eq. \eqref{eq:by_123t_G} and Eq. \eqref{eq:JEnt123} in Appendix.
Since $\bx||\bz$ is the ME estimator for $\bx$ using the past values of $\bz$ {and $\bx^{\rm JEnt} || \bz$ involves entropy minimization of the joint process $(\bx; \by)$ based on a joint model of $(\bx; \by;\bz)$}, 
it holds that
\begin{align}\label{eq:ineq1b}
h( \bx^{\rm JEnt} || \bz) \geq h(\bx || \bz).
\end{align}
Combining Eq. \eqref{eq:ineq1a} and Eq. \eqref{eq:ineq1b} leads to 
\begin{align}\label{eq:ineq1}
    \F^{\JEnt}_{\by \rightarrow \bx \mid \bz}&= 2(h(\bx^{\rm JEnt}  || \bz)-h(\bx^{\rm JEnt}  || \by\bz ))\nonumber \\
    &\geq 2(h(\bx || \bz)-h( \bx^{\rm JEnt}  || \by\bz )) \nonumber \\
    &=\F^{\Std}_{\by \rightarrow \bx \mid \bz}.
\end{align}
Furthermore, since $\bx||\by\bz$ is the ME estimator for $\bx$ using the joint process $\by$ and $\bz$, it holds that
\begin{align}
  h(\hat \bx || \by\bz) \geq h(\bx|| \by\bz).  
\end{align}
Therefore, the following inequality holds
\begin{align}\label{eq:ineq2}
    \F^{\Std}_{\by \rightarrow \bx \mid \bz} &= 2(h(\bx|| \bz)-h(\bx||\by\bz))\nonumber\\
    &\geq 2(h(\bx||\bz)-h(\hat\bx||\by\bz))\nonumber\\
   &= \F^{\SEnt}_{\by \rightarrow \bx \mid \bz}.
\end{align}
Moreover, by the definition in \eqref{eq:cGCMSEnt}, $\F^{\SEnt}_{\by \rightarrow \bx \mid \bz}\geq0$ which completes the proof.
\end{proof}

The differences between $\F^{\JEnt}_{\by \rightarrow \bx \mid \bz}, \F^{\Std}_{\by \rightarrow \bx \mid \bz}, \F^{\SEnt}_{\by \rightarrow \bx \mid \bz}$ reflect the distinctive features of ME estimation compared to orthogonal projection based MMSE estimation.

\begin{prop}\label{prop2}
Consider the VAR model for joint process $(\bx;
\by; \bz)$ in \eqref{eq:G_abc}, then $\F^{\JEnt}_{\by \rightarrow \bx \mid \bz}= \F^{\Std}_{\by \rightarrow \bx \mid \bz}= \F^{\SEnt}_{\by \rightarrow \bx \mid \bz}=0$ if and only if $G_{xy}(L)=0$.
\end{prop}

\begin{proof}

Based on the definition, it is straightforward to derive that if $G_{xy}(L) =0$ then 
\[
\bx^{\rm JEnt}||\bz = G_{xx}^{-1}(L) \beps_x,
\]
which is equal to $\bx^{\rm JEnt}||\by\bz$, see Eq. \eqref{eq:Gtildey} and Eq. \eqref{eq:JEnt123} in Appendix for more details about the definitions.
Then, $\F^{\JEnt}_{\by \rightarrow \bx \mid \bz}=0$. 
Thus, Proposition \ref{prop1} indicates that $\F^{\JEnt}_{\by \rightarrow \bx \mid \bz}= \F^{\Std}_{\by \rightarrow \bx \mid \bz}= \F^{\SEnt}_{\by \rightarrow \bx \mid \bz}=0$. It remains to show that if $G_{xy}(L)\neq0$, then 
\begin{align}\label{eq:prop2eq2}
 \F^{\SEnt}_{\by \rightarrow \bx \mid \bz}>0.
\end{align}
For this purpose, consider the following representation for the joint processes $\left(\bx;  \bz \right)$ and 
$\left(\by; \bz \right)$
\begin{align}
\tilde G_{xx}(L) \bx+ \tilde G_{xz}(L)\bz &= \tilde \beps_{x},\label{eq:hateps}\\
\check G_{yy}(L) \by + \check G_{yz}(L) \bz&= \check \beps_{y}.
\end{align}
Note that $\bx||\bz= \tilde G_{xx}(L)^{-1} \tilde \beps_{x}$ and $\by||\bz= \check G_{yy}(L) ^{-1}\check \beps_{y}$. 
Proposition \ref{prop3} shows that $\F^{\SEnt}_{\by \rightarrow \bx \mid \bz}=\F_{\check \beps_y \rightarrow \tilde \beps_x}$.
Therefore, it is equivalent to prove that if $G_{xy}(L)\neq 0$ then $ \F_{\check \beps_y \rightarrow \tilde \beps_x}>0$.

Assume $ \F_{\check \beps_y \rightarrow \tilde \beps_x}=0$. Then it indicates that $\tilde \beps_{x}$ is also the innovation process for the joint modeling  of $(\tilde \beps_{x}; \check \beps_{y})$ since the past values of  $\check \beps_{y}$ do not improve the prediction of $\tilde \beps_{x}$. Therefore, $\tilde \beps_{x} \perp H\{\check \epsilon_{y,t-k}, k \in \mZ^+  \}$ where $H\{\check \beps_{y,t-k}, k \in \mZ^+  \}$ denotes the Hilbert space spanned by all the past variables of $\check \beps_{y}$.
By definition $\tilde \beps_{x}$ satisfies that $\tilde \beps_{x} \perp H\{ \bz_{t-k}, k \in \mZ^+  \}$ since $\tilde \beps_{x}$ is the innovation process in Eq. \eqref{eq:hateps}. Then, 
$\cE(F(L) \hat \beps_{y} \tilde \beps_{x}^T) =\cE(F(L) \hat G_{yy}(L) \by \tilde \beps_{x}^T) =0$
for any causal filter $F(L)$.
Since $\hat G_{yy}(L)$ is invertible by assumption, then $\tilde \beps_{x} \perp H\{ \by_{t-k}, k \in \mZ^+  \}$.
Thus, $\tilde \beps_{x}  \perp H\{ \bx_{t-k}, \by_{t-k}, \bz_{t-k}, k \in \mZ^+  \}$ which indicates that $\tilde \beps_{x}= \beps_{x}$. Thus, Eq. \eqref{eq:hateps} implies that the optimal $G_{xy}(L)=0$ which contradicts to the assumption that $G_{xy}(L)\neq 0$. Therefore, $ \F_{\check \beps_y \rightarrow \tilde \beps_x}>0$ which completes the proof.
\end{proof}

Proposition \ref{prop2} indicates that if the joint model for $(\bx; \by; \bz)$ is estimated correctly, then all three cGCM measures do not lead to false positive connections. In practice, model parameters are estimated based on noisy and finite measurements, leading to noisy causality measures. In this case, causal inference is usually achieved by comparing the estimated measures with the probability distribution of measures under the null hypothesis. 
In this case, the three cGCM measures can potentially provide different sensitivity and specificity to detect network connections in statistical testing which are examined in the following section.

\begin{remark} In \cite{Geweke1982}, Geweke introduced the measure of instantaneous linear feedback $\F_{\bx\cdot \by}$ and the measure of linear dependence $\F_{\bx, \by}$ which have the following ME-based representations 
\begin{align*}
\F_{\bx\cdot \by}&= 2(h(\bx||\by) + h(\by||\bx) -h((\bx;\by)))\\
&=\ln \frac{\det \hat \Omega_{xx} \det \hat\Omega_{yy}}{\det\hat \Omega},\\
\F_{\bx, \by}&= 2(h(\bx) + h(\by) -h((\bx;\by)))\\
&=\ln \frac{\det \Sigma_{xx} \det \Sigma_{yy}}{\det\hat \Omega},
\end{align*}
which satisfy that
$\F_{\bx,\by}= \F_{\by\rightarrow \bx}+ \F_{\bx\rightarrow \by}+ \F_{\bx\cdot\by}$.
In analogy to Eq. \eqref{eq:fGCM}, the ME-based frequency-domain $\F_{\bx\cdot \by}$ and $\F_{\bx\cdot \by}$ can be constructed as
\begin{align*}
\f_{\bx\cdot \by}(\theta)&=\ln \frac{\det S_{\bx||\by}(\theta)  \det S_{\by||\bx}(\theta) }{\det S_{(\bx;\by)}(\theta)},\\
\f_{\bx, \by}(\theta)&=\ln \frac{\det S_\bx(\theta) \det S_{\by}(\theta)}{\det S_{(\bx;\by)}(\theta)}.
\end{align*}
In \cite{Geweke1984}, $\F_{\bx\cdot \by}$ and $\F_{\bx, \by}$ are generalized to the conditional linear feedback and linear dependence measures whose ME-based formulations can be constructed as
\begin{align}
\F_{\bx\cdot \by|\bz}^{\rm Std}&= 2(h(\bx||\by\bz) + h(\by||\bx\bz) -h((\bx;\by)||\bz))\label{eq:Fxyz_A1}\\
&=\ln \frac{\det \Omega_{xx} \det \Omega_{yy}}{\det \left[\begin{matrix}\Omega_{xx} & \Omega_{xy}\\ \Omega_{yx} & \Omega_{yy} \end{matrix} \right]},\label{eq:Fxyz_A2}\\
\F_{\bx, \by|\bz}^{\rm Std}&= 2(h(\bx||\bz) + h(\by||\bz) -h((\bx;\by)||\bz))\label{eq:Fxyz_B1}\\
&=\ln \frac{\det \tilde \Omega_{xx} \det \check \Omega_{yy}}{\det  \left[\begin{matrix}\Omega_{xx} & \Omega_{xy}\\ \Omega_{yx} & \Omega_{yy} \end{matrix} \right]}.\label{eq:Fxyz_B2}
\end{align}
Then the following decomposition holds
\begin{align*}
\F_{\bx,\by|\bz}= \F_{\by\rightarrow \bx|\bz}^{\rm Std}+ \F_{\bx\rightarrow \by|\bz}^{\rm Std}+ \F_{\bx\cdot\by|\bz}.
\end{align*}
The ME-based frequency-domain generalization of $\F_{\bx\cdot \by|\bz}^{\rm Std}$ and $\F_{\bx, \by|\bz}^{\rm Std}$ can be constructed straightforwardly by replacing the covariance matrices in \eqref{eq:Fxyz_A2}
and \eqref{eq:Fxyz_B2} by the PSD functions of the corresponding ME processes. Moreover, the SEnt-based formulations can be constructed analogously by replacing $\bx||\by\bz, \by||\bx\bz$ in \eqref{eq:Fxyz_A1} with $\bx^{\rm SEnt}||\by\bz$ and $\by^{\rm SEnt}||\bx\bz$. Similarly, the JEnt-based formulations can be obtained by replacing $\bx||\bz, \by||\bz, \bx||\by\bz, \by||\bx\bz$ with $ \bx^{\rm JEnt} ||\bz, \by^{\rm JEnt} ||\bz, \bx^{\rm JEnt} ||\by\bz, \by^{\rm JEnt} ||\bx\bz$, respectively.
\end{remark}

\section{Experiments}\label{sec:experiment}

\subsection{Simulations}

The sensitivity and specificity of the proposed methods and the standard GCM, cGCM methods have been compared using simulation data based on three different network structures.
The simulated time series were generated based on three VAR models of the following form
\begin{align}
   \bu_{t+1} = A\bu_t + \beps_t,
\end{align}
where the covariance of the process $\beps$ had a compound symmetry structure given by
\begin{align}
\cE(\beps_{i}\beps_{j})= 
\begin{cases}
1 &\text{if $i\neq j$,}\\
2 &\text{if $i = j$,}
\end{cases}
\end{align}
to represent correlated innovation processes. Three sets of system matrices $A$ were simulated according to the graphs illustrated in Figs. \ref{fig:simulation} (a), (b), (c), respectively, where the entry $A_{ij}=a$ if there is a directed connection from node $j$ to node $i$, including the case when $i=j$.
The graph in Fig. 2(a) has the same topology as an example used in the MVGC Matlab toolbox \cite{Barnett2014}. 
Fig. \ref{fig:simulation}(b) has a circular structure. Fig. \ref{fig:simulation}(c) illustrates a star-shaped graph where all nodes on the circular boundary receive input from the center node.
 For each VAR model, two sets of experiments were simulated with $a$ being chosen such as the maximum magnitude of the eigenvalues of $A$ were equal to $\lambda_{\rm max}= 0.85, 0.6$, respectively, to analyze the effect of different noise levels.

In the experiments, 1000 independent simulations were generated based on the VAR models with the length of observation being 120 for Figs. \ref{fig:simulation}(a) and (b) and 60 for Fig. \ref{fig:simulation}(c).
The simulated data were used to estimate the VAR model parameters using least-square fitting methods.
The VAR models were further transformed to state-space representations that were applied to compute the GCM and cGCM measures using the algorithms provided in the Appendix.

{To generate the null distributions of the estimated measures for statistical hypothesis testing, the time indices of each time series were randomly permuted while keeping other time series unchanged.}
Then, the VAR model parameters were estimated using the permuted time series and the GCM and cGCM measures were computed based on the estimated parameters. 1000 random permutations were applied to generate the null distributions for each pair of GCM or cGCM measures. Thus, 36000 (i.e., $1000\times 9\times 8/2$) permutations were generated in each trial.
A connection was declared if the p-value of the non-permuted measures in the null distribution was lower than a significance level that is determined with or without using methods for multiple comparisons. 
1000 p-values were collected from all the trials for each causality which were used to estimate the corresponding false-positive rate (FPR) and true-positive rate (TPR).

\subsection{In vivo MRI analysis}

The performance of these ME-based cGCM algorithms and standard methods have been compared using \emph{in vivo} MRI data from 100 unrelated subjects of the Human Connectome Project \cite{Glasser2013,Smith2013}.
Each subject has 4 rsfMRI scans with 2 mm isotropic voxels, matrix size = $104\times 90$, ${\rm TE=33.1 ms}$ and ${\rm TR = 0.72 s}$ and have been processed by the minimal processing pipeline \cite{Smith2013}.
Moreover, each subject has two diffusion MRI (dMRI) scans that were registered with the $T_{1w}$ MRI data.
The dMRI data has 1.25 mm isotropic voxels with the matrix size being $210\times 180$ and 111 slices.
Moreover, the dMRI data includes 3 non-zero b-values at $b=1000, 2000$, and 3000 ${ \rm s/mm^2}$ with ${\rm TE=89\, ms}$, and ${\rm TR = 5.5\, s}$ \cite{Sotiropoulos2013}. 

The Automated Anatomical Labeling (AAL) atlas \cite{Tzourio-Mazoyer2002} was applied to separate the brain and cerebellum into 120 regions. Then, the mean rsfMRI time series from each brain region is extracted and further processed to remove the mean signal and normalize the standard deviation.
The proposed time and frequency-domain GCM and cGCM methods were applied to analyze the 120-dimensional rsfMRI data, similar to the frequency-domain brain connectivity analysis used in \cite{6905833}.
Both the Akaike information criterion (AIC) and the Bayesian information criterion (BIC) \cite{McQuarrie1998} have been used for model selection. A first-order VAR model was optimal based on both AIC and BIC. The estimated VAR model parameters were applied to estimate the GCM and cGCM measures using the algorithms provided in the Appendix. 
The cGCM measure between two brain regions was computed conditional on the other 118-dimensional time series from all other regions whereas the GCM measures were estimated based on a linear model of the two-dimensional time series. 
All the measures were computed based on the same VAR model of the 120-dimensional time series to ensure the comparisons were consistent as suggested in \cite{Dhamala2018,Barnett2018,Barnett2015}. 

For human brains, the ground true effective connections between brain regions are unknown. But the structural connectivity (SC) via white matter pathways provides the biological substrates for brain effective connectivity \cite{Park2013,Sokolov2018,Sokolov2019}. 
To this end, the multi-fiber tractography algorithm developed in \cite{Malcolm2010,Reddy2016} was applied to the diffusion MRI data to estimate the whole brain fiber bundles. 
Then, the anatomically curated white-matter atlas \cite{Zhang2018} was applied to the tractography results to filter out possibly false connections. Next, the percentage of fiber bundles between each pair of ROIs was computed which was considered as the weight of the SC.

\section{Results}\label{sec:results}

\subsection{Results on simulation experiments}

The second row of Fig. \ref{fig:simulation} illustrates the sample mean fGCM or fcGCM functions of all non-zero connections corresponding to the three network structures illustrated in the first row of Fig. \ref{fig:simulation} with $\lambda_{\rm max}=0.85$. The ME-based fGCM and fcGCM functions have negative values at high-frequency range on average. On the other hand, the original fGCM and fcGCM functions are all positive.
The third row of Fig. \ref{fig:simulation} 
illustrates the receiver operating characteristic (ROC) curves, i.e. TPR vs FPR, when the average values of fGCM and fcGCM functions at $[0~\tfrac\pi2]$ frequency interval were used to detect connections. The abbreviations cGCM-Std, cGCM-SEnt, and cGCM-JEnt represent the standard cGCM Eq. \eqref{eq:cGCMstd}, separate and joint ME estimation-based cGCMs, respectively.
In all three figures, cGCM-SEnt has the highest TPR if FPR is controlled at the same value.

The last row of Fig. \ref{fig:simulation} shows the ROC curve based on the time-domain GCM and cGCM measures. Though cGCM-SEnt still has relative higher accuracy than the other two cGCM methods, their performance were more similar compared to the results in the second row of Fig. \ref{fig:simulation}. 
Moreover, the frequency-domain cGCM-SEnt based results in Figs. \ref{fig:simulation}g and \ref{fig:simulation}h have shown better performance compared to the time-domain measure-based results in Figs. \ref{fig:simulation}j and \ref{fig:simulation}k. 

The Supplementary Materials provide additional results based on different multiple-comparison methods
and the results corresponding to VAR models with $\lambda_{\rm max}=0.6$. All results indicate that frequency-domain cGCM-SEnt can provide the highest TPR with suitable control of the FPR.

\begin{figure}[tbhp]
\centering
\includegraphics[width=1\linewidth]{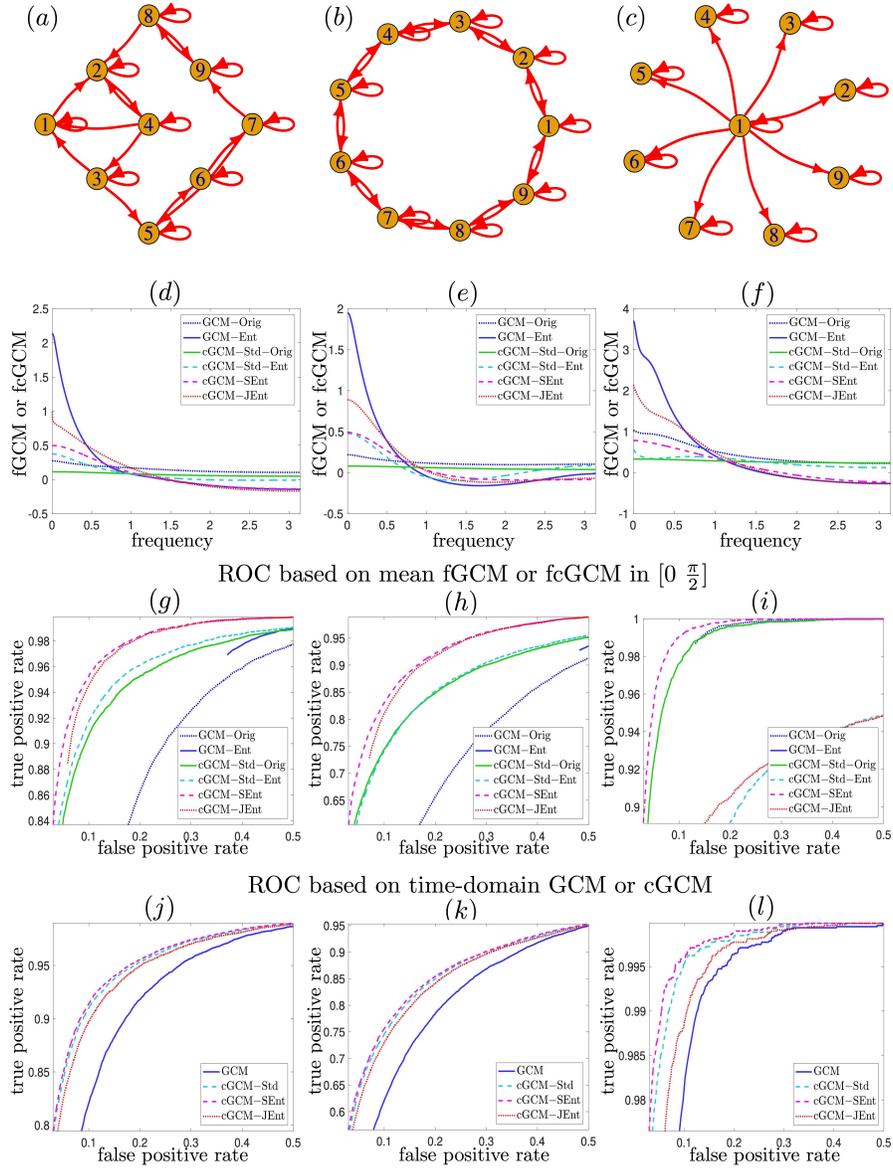}
\caption{Illustration of simulation results corresponding to VAR models with $\lambda_{\rm max} =0.85$. The first row demonstrates the structure of three VAR models used in the simulations. The second row shows sample mean of fGCM and fcGCM functions for all non-zero connections in the first row.
The third row shows the ROC curves based on mean fGCM and fcGCM values in the frequency interval $[0~\tfrac\pi 2]$. The last row illustrates the ROC curve based on time-domain GCM and cGCM measures.
}
\label{fig:simulation}
\end{figure}

\subsection{On correlation with SC}

The bar plots in Fig. \ref{fig:corr} illustrate the correlation coefficients between the SC and different GCM and cGCM measures among 100 subjects.
The first three plots show the correlation coefficients between SC and the GCM, the ME-based fGCM-Ent, and the original fGCM-Geweke method, respectively. 
fGCM-Ent has a significantly higher ($p<10^{-60}$, t-test) correlation with SC than the other two methods.
 On the other hand, the original fGCM-Geweke has a significantly lower correlation with SC than the time-domain GCM ($p<10^{-60}$, t-test).

All three ME-based fcGCM measures have a significantly higher correlation with SC than the corresponding time-domain measures ($p<10^{-40}$, t-test). Moreover, fcGCM-Std-Ent has a significantly higher correlation than the original method fcGCM-Std-Geweke ($p<10^{-50}$, t-test). Thus, the ME-based fcGCM measures can more correctly reflect the passband of the hemodynamic response in the rsfMRI data. Among the three ME-based fcGCM, fcGCM-SEnt has a significantly higher ($p<10^{-7}$, t-test) correlation with SC than the other two methods. More details of the estimation results are provided in the Supplementary Materials.

\begin{figure}
\centering
\includegraphics[width=1\linewidth]{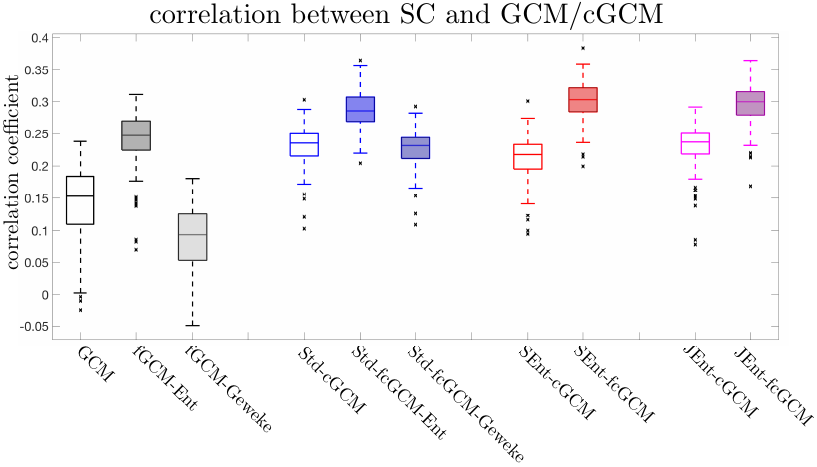}
\caption{The correlation coefficients between SC and different GCM and cGCM measures in whole-brain networks among 100 subjects.}
\label{fig:corr}
\end{figure}

\subsection{Comparison of fGCM and fcGCM}

To further compare the fGCM and fcGCM methods, Fig. \ref{fig:freq} shows these frequency-domain measures for two pairs of brain regions.
On the left panel, Figs. \ref{fig:freq}(c) and (e) illustrate the fGCM and fcGCM functions for the connection from left angular (Angular-L) to the left frontal superior medial cortex (Frontal-Sup-Medial-L), which are both involved in the default mode network\cite{Oliver2019}.
{The dashed lines illustrate the mean values among 100 subjects and the error bars show the range of standard deviations. The constant solid lines illustrate the average values of these functions which are equal to the corresponding time-domain measures. It is noted that the two functions have the same mean.}
In Fig. \ref{fig:freq}(c), fGCM-Ent is significantly higher than fGCM-Geweke in the hemodynamic-related 0.01 and 0.1 Hz range.
Fig. \ref{fig:freq}(e) shows the four fcGCM functions of the Angular-L to Frontal-Sup-Medial-L connection, where
the three ME-based methods, including fcGCM-Std-Ent, fcGCM-SEnt, and fcGCM-JEnt, all have significantly higher values than the original fcGCM-Std-Geweke method between 0.01 and 0.1 Hz. On the other hand, fcGCM-Std-Geweke shows weaker contrast between different frequencies than other ME-based fcGCM methods. {The constant solid lines show the corresponding time-domain measures. The black box shows more details about the difference between the mean values, where fcGCM-JEnt and fcGCM-SEnt have the highest and the lowest values, respectively, which is consistent with Proposition 1.}

The right panel in Fig. \ref{fig:freq} shows the fGCM and fcGCM functions for the connection from the right angular (Angular-R) to the left frontal superior medial cortex (Frontal-Sup-Medial-L). Since there do not exist anatomically plausible direct structural connections between the two regions, the GCM and cGCM measures are expected to have lower values than the results on the left panel.
As expected, the peak magnitudes of fGCM-Ent and fGCM-Geweke in Fig. \ref{fig:freq}(d) are only 28\% and 30\% of the corresponding peak magnitudes shown in Fig. \ref{fig:freq}(c).
In Fig. \ref{fig:freq}(f), the peak magnitudes of fcGCM-Std-Ent, fcGCM-Std-Geweke, fcGCM-SEnt and fcGCM-JEnt are reduced to 3\%, 11\%, 9\%, and 12\% of the corresponding peak magnitudes in Fig. \ref{fig:freq}(e). More significant reductions in cGCM measures indicate that cGCM has better performance in reducing false connections than GCM.

\begin{figure}
\centering
\includegraphics[width=.8\linewidth]{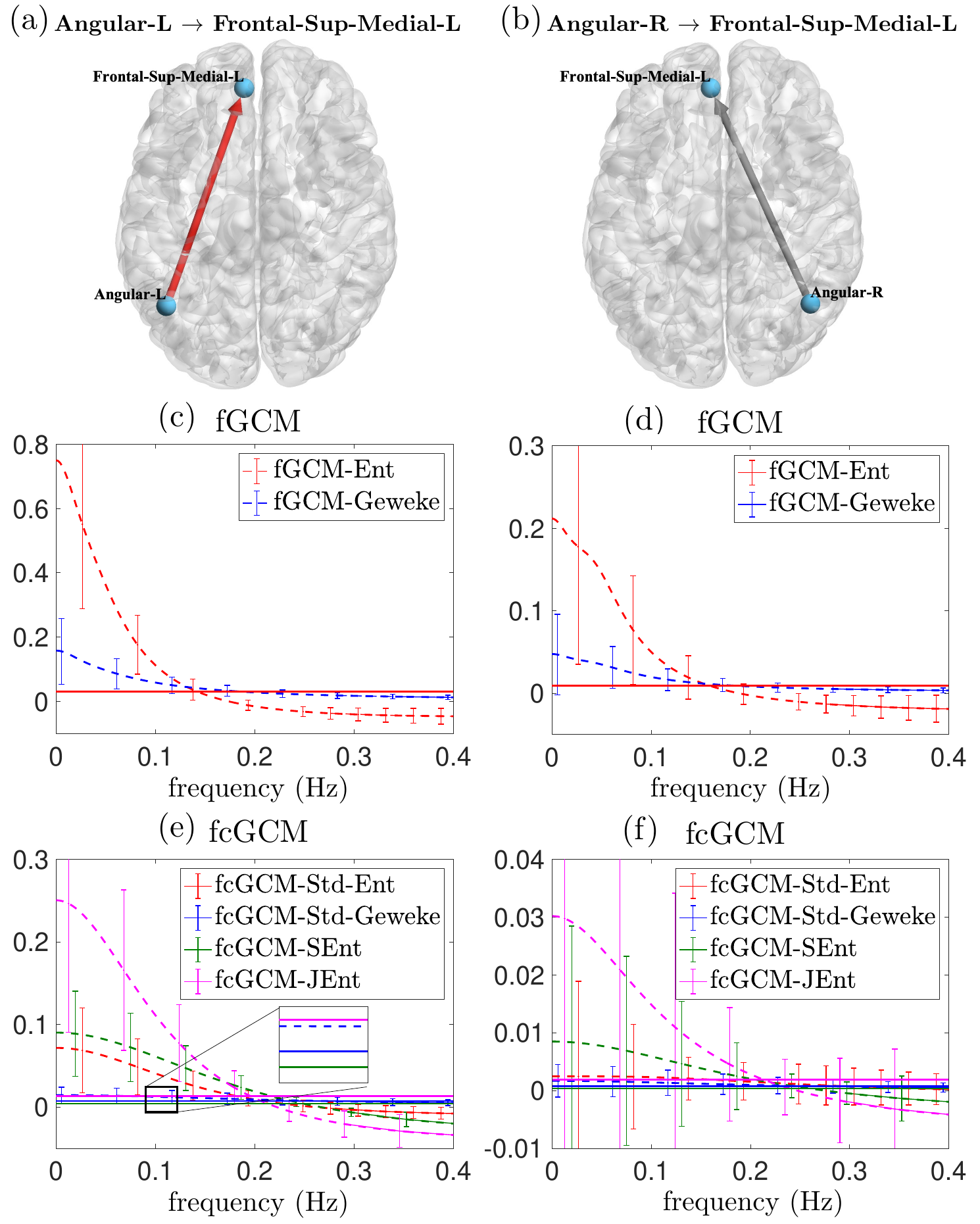}
\caption{Comparison of the fGCM and fcGCM methods on two brain connections. (a) and (b) illustrate two brain connections. (c) and (d) show the comparison of the results of two fGCM methods. (d) and (e) demonstrate the results of different fcGCM methods. The mean values of the fGCM and fcGCM measures are shown in solid lines.}
\label{fig:freq}
\end{figure}

\section{Discussion}\label{sec:discussion}

Experimental results based on simulations have shown that frequency-domain fcGCM-SEnt has enhanced sensitivity and specificity in detecting network connections compared to other frequency-domain measures with a suitable choice of significance level.
The performance of fcGCM-SEnt is also better than the corresponding time-domain method for the first two networks relatively more complex structures. For the star-shaped structure, all time-domain methods have provided high accuracy. 
On the other hand, simulation results provided in Supplementary Materials have shown that time-domain cGCM measures may have lower accuracy than GCM in high-noise situations. But frequency-domain cGCM measures such as fcGCM-SEnt and fcGCM-JEnt still provide the highest level of sensitivity and specificity in this situation.
Moreover, results based on \emph{in vivo} rsfMRI data have shown consistent results that fcGCM-SEnt measures of human brain networks have provided the highest correlation with SC measures obtained using diffusion MRI tractography compared to other measures.  
Furthermore, ME-based fGCM and fcGCM can correctly characterize the bandpass property in the hemodynamic response in rsfMRI signals that cannot be achieved by using the standard methods. Thus, the proposed method can potentially be a useful tool in neuroimaging analysis. The toolbox used to compute the proposed measures is available at the GitHub \url{https://github.com/LipengNing/ME-GCM}.

Below are some discussions that highlight the limitations and future work.

The proposed method is limited by the assumptions of linear models and stationary Gaussian processes. Linear systems cannot accurately model biological measurements \cite{Jirsa2002,Valdes-Sosa2009,Valdes-Sosa2011,Friston2013} or time-varying dynamics in brain networks\cite{Liu2013,Zalesky2014,Ning2019,Ning2020a}. Thus, generalizations of the ME method for nonlinear, non-Gaussian, or nonstationary processes and the nonparametric method for ME estimation will need to be derived in future work.

The main goal of this work is to derive an information-theoretic-based framework for conditional causal inference between multivariate time series. For this purpose, we have reviewed the GCM and cGCM and several existing information-theoretic-based formulations such as the transfer entropy (TE) and directed information (DI) and derived three ME-based measures, and analyzed the relationship between them. 
It should be noted that there are several other methods derived for causal inference including the modified GCM method \cite{Hosoya2017}, the generalized partial directed coherence \cite{4288544} and the directed transfer function \cite{KORZENIEWSKA2003195}. A comprehensive comparison of these methods is beyond the scope of this paper and will be explored in future work.

The proposed solutions for the ME filters are derived based on the assumption that the diagonal entries of the VAR representations, e.g., $\hat G_{xx}(L)$ in Eq. \eqref{eq:ME_filter}, are stably invertible. Though this assumption is typically satisfied in the proposed experiments, rigorous solutions for the optimal filter when the assumption does not hold will need to be derived in future work. 
In this case, it is a hypothesis that the optimal causal filter cannot reduce the entropy to the same level as the unstable or noncausal filter. As a result, the corresponding ME-based GCM in Eq. \eqref{eq:GCM} is lower than the original GCM. 

The rsfMRI experiment has ignored the hemodynamic response functions (HRFs) and measurement noise. It was shown in \cite{10.1162/NECO_a_00828} that the HRFs are non-minimum-phase (NMP) and spatially varying NMP HRFs can distort GCM measures. Thus, heterogeneous HRFs are confounding factors in the correlation analysis with structural connectivity. Moreover, it was also shown in \cite{10.1162/NECO_a_00828} that the distortions of GCM cannot be resolved by using standard Wiener deconvolution filters. Thus, the physiological meaning of the estimated GCM and cGCM measures needs to be further validated. It should be noted that the proposed methods can also be applied to analyze other types of imaging signals such as EEG and MEG, though these modalities may reduce the spatial resolution and limit the correlation analysis with structural connectivity. Further analysis is needed to examine the performance of the proposed method for EEG and MEG.

The meaning and usefulness of negative fGCM and fcGCM based on ME methods need to be further examined and understood. As illustrated in the experiments, ME-based fGCM and fcGCM methods may have negative values which is a major difference from the original methods in \cite{Geweke1982,Geweke1984}. The original fGCM and fcGCM measures were heuristically constructed to satisfy Geweke's requirements that the frequency-domain measures should be nonnegative and their mean value should be equal to the corresponding time-domain measures. On the other hand, the ME-filter minimizes the entropy of the residual processes but may not be able to reduce the power in all frequencies. It may selectively increase the power at some less important frequencies to ensure that the overall entropy rate is minimized. Experimental results based on simulations and \emph{in vivo} rsfMRI data have shown that negative fGCM and fcGCM values typically occur at high frequencies range beyond the passband of the underlying signals. Thus, the sign of ME-based fGCM and fcGCM measures may be sensitive to the intrinsic frequency of the underlying signals. More validations will be implemented in future wor.

{
Finally, it is noted that the standard and separate ME-based cGCM measures require the spectral factorization algorithm for each pair of signals. Although efficient iteration algorithms with an $o(n)$ complexity \cite{Chu2015} have been developed to solve the DARE to obtain spectral factorizations, the number of measures scales as $o(n^2)$ for large-scale networks. Moreover, the computational cost is further increased if random permutation tests are used for statistical inference. Thus, computationally efficient algorithms and statistical analysis methods need to be developed in future work.}

\section{Conclusions}\label{sec:conclusions}
The proposed ME-based methods for conditional causality analysis, especially cGCM-SEnt, provide more sensitive and specific measures than the standard methods in detecting network connections. The ME-based frequency-domain methods can correctly characterize the band-pass property of brain activities using neuroimaging data. Thus, the ME-based method can provide more effective tools for causal inference of direct connections in networks which may be useful tools for neuroimaging analysis. Future work will focus on the derivation of the general solution for ME filters without the assumption of stable diagonal blocks of VAR models and the development of efficient computation algorithms and statistical analysis methods. Moreover, the integration of structural constraints in model estimation will be explored in future work to reduce redundant parameters and improve the reliability of the proposed measures for large-scale time series analysis. Furthermore, nonlinear generalization of the proposed measures as in \cite{Sun2008,6104226,doi:10.1098/rsos.200863,Marinazzo2011} will also be explored in future work.

\appendix

\section{Consistent models for time series using tate-space representations}\label{sec:appendix_1}

This subsection introduces some preliminary results on state-space representation (SSR) for multivariate which are useful to derive models of the subprocess for computing the GCM and cGCM. 
Assume that the multivariate time series $\bu\in \mR^n$ is modeled by the following SSR of the innovation form
\begin{subequations}\label{eq:SS_B}
\begin{align}
\bxi_{t+1} &= A \bxi_{t} + B \beps_t, \\
\bu_t &= C\bxi_t + \beps_t,
\end{align}
\end{subequations}
where $A\in \mR^{m\times m}, B\in \mR^{m\times n}, C\in \mR^{n\times m}$ with $m\geq n$ and $\beps\in \mR^n$ represents the zero-mean white Gaussian innovation process with $\cE(\beps \beps^T) =\Omega$. A general SSR can be transformed to the innovation form in Eq. \eqref{eq:SS_B}
by using the spectral factorization algorithm \cite{Kailath2000}.
Then, the power spectral density (PSD) function of $\bu_t$ can be expressed as
\begin{align}\label{eq:spectfact_general}
S(\theta) = H(e^{i\theta}) \Omega H(e^{i\theta}) ^*,
\end{align}
 where $i=\sqrt{-1}$, and 
 \begin{align}
     H(L) = I_n+C(I_m-AL)^{-1} BL
 \end{align} 
 represents the transfer function from $\beps$ to $\bu$ such that 
 \begin{align}\label{eq:Hyeps}
 \bu = H(L) \beps.
 \end{align}
Moreover, $H(L)$ is a minimum-phase transfer function implying all its zeros are outside of the unit circle.
For convenience, $H(L)$ is represented by
\begin{align}
H(L) \sim
\left(\begin{array}{c|c}
  A & B \\
  \hline
  C & I_n
    \end{array} \right).
\end{align}

The inverse representation of Eq. \eqref{eq:Hyeps} is given by the following VAR representation
\begin{align}
G(L)\bu = \beps,\label{eq:tf_zy},
\end{align}
with $G(L)$ is given by
\begin{align}\label{eq:HinvG}
G(L) \sim \left(\begin{array}{c|c}
  A-BC & B \\
  \hline
  -C & I_n
    \end{array} \right).
\end{align}
If $H(L)$ is a minimum-phase spectral factorization, then $A-BC$ is stable, i.e. the eigenvalues of $A-BC$ are within the unit circle. 
Throughout the paper, all the diagonal blocks of $G(L)$ and other VAR filters are assumed to have stable inversion.
Specifically, it is assumed that $A- \sum_{k\in \cK}b_kc_k^T$ is stable for all subsets $\cK \subseteq \{1, \ldots, n\}$ with $b_k$ and $c_k$ being the k-th columns of $B$ and $C^T$, respectively. 

Based on Eq. \eqref{eq:Hyeps},  the joint process $\bx$ and $\by$ is represented by
\begin{align}\label{eq:xy_sub}
\left[ \begin{matrix} \bx\\ \by \end{matrix} \right] = [I_{n_x+n_y}\; 0] H(L) \beps.
\end{align}
Then the spectral factorization algorithm \cite{Kailath2000} can be applied to the above equation to obtain the following representation
\begin{align}\label{eq:xy_H}
\left[ \begin{matrix} \bx\\ \by \end{matrix} \right] = \hat H(L)\left[ \begin{matrix} \hat \beps_{x} \\ \hat \beps_{y} \end{matrix} \right],
\end{align}
where $\hat H(L)$ is related to $\hat G(L)$ via Eq. \eqref{eq:hatHL}.
{The spectral factorization algorithm requires the solution to a discrete-time algebraic Riccati equation (DARE) which can be obtained by using an iteration algorithm with an $o(n)$ complexity in each iteration \cite{Chu2015} for large-scale systems.}
For the ME process $\bx||\by$ introduced in Eq. \eqref{eq:y12},  the corresponding PSD function given by
\begin{align}\label{eq:S1c2}
S_{\bx||\by}(\theta) = \hat G_{xx}(e^{i\theta})^{-1} \hat \Omega_{xx} \hat G_{xx}(e^{i\theta})^{-*},
\end{align} 
where $X^{-*}=(X^{-1})^*$.

The process $\bx$ can be represented by 
\begin{align}
\bx = [I_{n_x}\; 0] H(L) \beps. 
\end{align}
The spectral factorization algorithm can be applied to compute the covariance matrix $\Sigma_{xx}$ of one-step-ahead prediction error $\bx \mid \X_{-}$.

\section{The original cGCM and ME-based fcGCM}\label{sec:appendix_3}

By using the same method as in Eqs. \eqref{eq:xy_sub}, \eqref{eq:xy_H}, the joint process $(\bx; \bz)$ is represented by
\begin{align}\label{eq:Gy13}
&\left[\begin{matrix} \tilde G_{xx}(L)& \tilde G_{xz}(L) \\ \tilde G_{zx}(L)& \tilde G_{zz}(L) \end{matrix}\right]  \left[\begin{matrix} \bx\\ \bz\end{matrix}\right]  = \left[\begin{matrix} \tilde \beps_{x} \\ \tilde \beps_{z}  \end{matrix}\right],
\end{align}
where
\begin{align}
\cE\left(\left[\begin{matrix} \tilde \beps_{x} \\  \tilde \beps_{z} \end{matrix}\right]
\left[\begin{matrix} \tilde \beps_{x}^T &  \tilde \beps_{z}^T \end{matrix}\right]
\right)
= \left[\begin{matrix} \tilde \Omega_{xx} & \tilde \Omega_{xz} \\ \tilde \Omega_{zx} & \tilde \Omega_{zz} \end{matrix} \right].
\end{align}
Then the ME process $\bx||\bz$ is given by
\begin{align}
\bx||\bz &= \bx +\tilde G_{xx}(L)^{-1} \tilde G_{xz}(L) \bz,\\
&= \tilde G_{xx}(L)^{-1} \tilde \beps_{x}, \label{eq:y1cond3}
\end{align}
with the corresponding PSD being equal to
\begin{align}\label{eq:s_xz}
S_{x||z}(\theta) =  \tilde G_{xx}(e^{i\theta}) ^{-1}\tilde \Omega_{xx} \tilde G_{xx}(e^{i\theta})^{-*} . 
\end{align}

Based on the joint VAR model for $(\bx;\by;\bz)$ in Eq. \eqref{eq:G_abc}, the ME process $\bx||\by\bz$ is represented by
\begin{align}\label{eq:by_123t_G}
\bx||\by\bz &= \bx+G_{xx}^{-1}(L)G_{xy}(L) \by+G^{-1}_{xx}(L)G_{xz}(L) \bz,\nonumber\\
&= G_{xx}^{-1}(L) \beps_{x}
\end{align}
with the corresponding PSD being equal to
\begin{align}\label{eq:S123}
S_{\bx||\by\bz}(\theta) =  G_{xx}^{-1}(e^{i\theta}) \Omega_{xx} G_{xx}^{-*}(e^{i\theta}).
\end{align}

The standard cGCM, i.e. cGCM-Std, can be expressed using the ME process as 
\begin{align}
\F^{\rm Std}_{\by \rightarrow \bx | \bz} &= 2(h(\bx|| \bz) - h(\bx||\by\bz)),\\
 &= \ln \frac{\det \tilde \Omega_{xx}}{\det \Omega_{xx}}.\label{eq:FOrig}
\end{align}
Then, \eqref{eq:s_xz} and \eqref{eq:S123} can be applied to \eqref{eq:fcGCMstd} to compute $\f^{\rm Std-Ent}_{\by\rightarrow \bx \mid \bz}(\theta)$

For comparison, the original fcGCM, i.e. fcGCM-Std-Geweke, is introduced below. 
First, combine Eq. \eqref{eq:Gy13} and Eq. \eqref{eq:Hyeps} to obtain the following equations
\begin{align}
 \left[\begin{matrix} \tilde \beps_{x} \\ \by \\ \tilde \beps_{z}  \end{matrix}\right] &=\left[  \begin{matrix} \tilde G_{xx}(L)& 0 & \tilde G_{xz}(L)\\ 0& I_{n_y}& 0\\  \tilde G_{zx}(L)& 0& \tilde G_{zz}(L)\end{matrix} \right]\nonumber \\
 & \times  \left[  \begin{matrix}H_{xx}(L)& H_{xy}(L)& H_{xz}(L)\\  H_{yx}(L)& H_{yy}(L)& H_{yz}(L)\\  H_{zx}(L)& H_{zy}(L)& H_{zz}(L)\end{matrix} \right]\left[\begin{matrix} \beps_{x}\\ \beps_{y}\\ \beps_{z}\end{matrix} \right],\\
  &=\left[  \begin{matrix}P_{xx}(L)& P_{xy}(L)& P_{xz}(L)\\  P_{yx}(L)& P_{yy}(L)& P_{yz}(L)\\  P_{zx}(L)& P_{zy}(L)& P_{zz}(L)\end{matrix} \right]\left[\begin{matrix} \beps_{x}\\ \beps_{y}\\ \beps_{z}\end{matrix} \right].
 \end{align}
Then, the original fcGCM can be computed as
\begin{align}\label{eq:Geweke_fcGCM}
&\f^{\rm Std-Geweke}_{\by \rightarrow \bx | \bz}(\theta)  = \f^{\rm Geweke}_{\tilde \beps_z \by \rightarrow \tilde\beps_x}(\theta) \\
& = \ln \frac{\det \tilde \Omega_{xx}}{\det (\tilde \Omega_{xx} - [P_{xy}(e^{i \theta})\, P_{xz}(e^{i \theta}) ] \Omega_{yz|x}  [P_{xy}(e^{i \theta})\, P_{xz}(e^{i \theta}) ]^* )},\label{eq:cGCM-Geweke}
\end{align}
where 
\begin{align}
\Omega_{yz|x} = \left[ \begin{matrix} \Omega_{yy}& \Omega_{yz} \\ \Omega_{zy} & \Omega_{zz} \end{matrix}\right] -\left[ \begin{matrix}  \Omega_{yx} \\ \Omega_{zx}\end{matrix} \right]\Omega_{xx}^{-1} \left[ \begin{matrix}  \Omega_{xy} & \Omega_{xz}\end{matrix} \right].
\end{align}
The mean value of $\f^{\rm Std{-}Geweke}_{\by \rightarrow \bx | \bz}(\theta)$ is equal to $\F^{\rm Std}_{\by \rightarrow \bx | \bz}$.

\section{Separate and joint ME-based cGCM}\label{sec:appendix_4}

This subsection introduces the computational algorithms for cGCM-SEnt and cGCM-JEnt in Figs. 1(d) and 1(e).
Similar to Eq. \eqref{eq:Gy13}, the joint process $(\by; \bz)$ is represented by
\begin{align} \label{eq:Gy23}               
&\left[\begin{matrix} \check G_{yy}(L)& \check G_{yz}(L) \\ \check G_{zy}(L)& \check G_{zz}(L) \end{matrix}\right]  \left[\begin{matrix} \by\\ \bz\end{matrix}\right]  = \left[\begin{matrix} \check \beps_{y} \\ \check \beps_{z}  \end{matrix}\right],
\end{align}
where
\begin{align}
\cE\left(\left[\begin{matrix} \check \beps_{y} \\  \check \beps_{z} \end{matrix}\right]
\left[\begin{matrix} \check \beps_{y}^T &  \check \beps_{z}^T \end{matrix}\right]
\right)
= \left[\begin{matrix} \check \Omega_{yy} & \check \Omega_{yz} \\ \check \Omega_{zy} & \check \Omega_{zz} \end{matrix} \right].
\end{align}
Then, the ME process $\by||\bz$ is given by
\begin{align}
\by||\bz &= \by + \check G_{yy}(L)^{-1} \check G_{yz}(L) \bz,\\
&= \check G_{yy}(L)^{-1} \check \beps_{y}. \label{eq:y2cond3}
\end{align}

By definition, cGCM-SEnt and fcGCM-SEnt are equal to
\begin{align}
\F_{\by \rightarrow \bx|| \bz}^{\rm SEnt}&= \F_{\by||\bz\rightarrow \bx||\bz},\\
\f_{\by \rightarrow \bx|| \bz}^{\rm SEnt}(\theta)&= \f_{\by||\bz\rightarrow \bx||\bz}(\theta),
\end{align}
respectively.
The following proposition is useful to derive the computation methods for $\F_{\by \rightarrow \bx|| \bz}^{\rm SEnt}$ and $\f_{\by \rightarrow \bx|| \bz}^{\rm SEnt}(\theta)$ and to prove Proposition 2.

\begin{prop}\label{prop3}
For $\tilde \beps_{x}$ and $\check \beps_{y}$ defined in Eq. \eqref{eq:Gy13} and Eq. \eqref{eq:Gy23}, respectively, The following equations hold.
\begin{align}
\F_{\by||\bz\rightarrow \bx||\bz}&= \F_{\check \beps_y \rightarrow \tilde \beps_x},\label{eq:SEnt}\\
\f_{\by||\bz \rightarrow\bx||\bz}(\theta) &=  \f_{\check \beps_y \rightarrow \tilde \beps_x}(\theta).\label{eq:fSEnt}
\end{align}
\end{prop}
\begin{proof}
Let $\hat F(L)$ denote the causal filter that minimizes the entropy of
\begin{align}
\hat \bx||\by\bz  = \bx||\bz - \hat F (L)  \by||\bz.
\end{align}
From Eq. \eqref{eq:y1cond3} and Eq. \eqref{eq:y2cond3}, 
\begin{align}
\hat \bx||\by\bz &= {\tilde G_{xx}(L)}^{-1}\big(\tilde \beps_{x} - \tilde G_{xx}(L)\hat F(L) {\check G_{yy}(L)}^{-1} \check \beps_{y}\big).
\end{align}
Therefore, the entropy of $\hat \bx||\by\bz$ is equal to that of $\tilde \beps_{x} - \tilde G_{xx}(L)\hat F(L) {\check G_{yy}(L)}^{-1} \check \beps_{y}$. 

On the other hand, if $\tilde F(L)$ is a causal filter that minimizes the entropy of 
\begin{align}
\tilde \beps_x|| \check \beps_{y} := \tilde \beps_{x}- \tilde F(L) \beps_{y},
\end{align}
then 
\begin{align}
\hat F(L) = {\tilde G_{xx}(L)}^{-1} \tilde F(L) \check G_{yy}(L)
\end{align}
also minimizes the entropy of $\hat \bx||\by\bz$. 
Therefore, 
\begin{align}
\hat \bx||\by\bz &= {\tilde G_{xx}(L)}^{-1} \tilde \beps_x|| \check \beps_{y},\label{eq:SEnt_y_mu}
\end{align}
which implies
\begin{align}
h(\hat \bx||\by\bz) & = h(\tilde \beps_x|| \check \beps_{y})\label{eq:SEnt_eta_mu}.
\end{align}
On the other hand, Eq. \eqref{eq:y1cond3} and Eq. \eqref{eq:y2cond3} imply that
\begin{align}\label{eq:y13_eta}
h(\bx||\bz) = h(\tilde \beps_x),
\end{align}
which proves Eq. \eqref{eq:SEnt}.

To prove Eq. \eqref{eq:fSEnt}, it is noted that
\begin{align}
\det S_{\bx||\bz}(\theta) &= \det( \tilde G_{xx}(e^{i\theta})^{-1} \tilde  G_{xx}(e^{i\theta})^{-*})  \det \tilde \Omega_{xx},\\
\det S_{\hat \bx|| \by\bz}(\theta) &=  \det( \tilde G_{xx}(e^{i\theta})^{-1} \tilde  G_{xx}(e^{i\theta})^{-*}) \det S_{\tilde \beps_x|| \check \beps_y}(\theta).
\end{align}
Therefore, the following equations hold.
\begin{align*}
\f_{\by||\bz\rightarrow \bx|| \bz}(\theta) &=\ln \frac{\det S_{\by||\bz}(\theta)}{\det S_{\hat \bx||\by\bz}(\theta)} \\
& = \ln \frac{ \det \tilde \Omega_{xx}}{ \det S_{\tilde \beps_x|| \check \beps_y}(\theta)}\\
&=  \f_{\check \beps_y \rightarrow \tilde \beps_x}(\theta),
\end{align*}
which completes the proof.
\end{proof}

Next, the SSR of the joint process $(\tilde \beps_x; \check \beps_y)$ is derived in order to compute $\F_{\check \beps_y \rightarrow \tilde \beps_x}$ and $\f_{\check \beps_y \rightarrow \tilde \beps_x}(\theta)$.
To this end, the following matrices are introduced
\begin{align}
B_a &= [\tilde B_x,0, \tilde B_z], \\
B_b &= [0, \check B_y, \check B_z],
\end{align}
where $\tilde B$ and $\check B$ are the corresponding input matrices of the SSR of the innovation form for $\tilde G(L)$ in Eq. \eqref{eq:Gy13} and $\check G(L)$ in Eq. \eqref{eq:Gy23}, respectively.
$\tilde B_x$ and $\tilde B_z$ represent the first $n_x$ and the last $n_z$ columns of $\tilde B$, $\check B_y$ and  $\check B_z$ represent the first $n_y$ and last $n_z$ columns of $\check B$, respectively. By multiplying the transfer function from $\bu_t$ to $\tilde \beps_{x}$ and $\check \beps_{y}$ and the transfer function from $\beps_t$ to $\bu_t$, the transfer function from $\beps_t$ to $\tilde \beps_{x}$ and $\check \beps_{y}$ can be computed as
\begin{align}\label{eq:Gaug}
H^{\rm aug}(L)\sim \left(\begin{array}{c|c}
   A^{\rm aug} & B^{\rm aug}\\
  \hline
 C^{\rm aug} & D^{\rm aug}
    \end{array} \right),
\end{align}
where
\begin{align}
A^{\rm aug} &={\footnotesize \left[\begin{matrix}A & 0 &0 \\ B_a C&A -\tilde B \tilde C & 0 \\
B_b C& 0 &  A -\check B\check C \end{matrix} \right]},\\
{B^{\rm aug}}&= \left[B;\, B_a;\, B_b \right],\\
C^{\rm aug}&= \left[\begin{matrix}C_x& -C_x & 0\\ C_y& -C_y& 0\end{matrix}\right],\\
D^{\rm aug}&= [I_{n_x+n_y} \, 0 ],
\end{align}
where $C_x$ and $C_y$ are the first $n_x$ rows and the following $n_y$ rows of the measurement matrix $C$.
Next, $\F_{\by||\bz\rightarrow \bx||\bz}$ and $\f_{\check \beps_y \rightarrow \tilde \beps_x}(\theta)$ can be computed by using the spectral factorization algorithm.

To compute the cGCM-JEnt measure illustrated in Fig.1(e), consider the following representation $(\bx^{\rm JEnt} || \bz;  \by^{\rm JEnt} || \bz)$
\begin{align}
\left[\begin{matrix}  \bx^{\rm JEnt} || \bz \\ \by^{\rm JEnt}  ||\bz\end{matrix} \right] &=  \left[  \begin{matrix}G_{xx}(L)& G_{xy}(L)\\  G_{yx}(L)& G_{yy}(L)\end{matrix} \right]^{-1}\left[\begin{matrix} \beps_{x} \\  \beps_{y} \end{matrix} \right].\label{eq:Gtildey}
\end{align}
The SSR of the above transfer function is given by
\begin{align}\label{eq:SS_xy_JEnt}
 \left[  \begin{matrix}G_{xx}(L)& G_{xy}(L)\\  G_{yx}(L)& G_{yy}(L)\end{matrix} \right]^{-1} \sim \left(\begin{array}{c|c}
   A-B_zC_z & [B_x\, B_y] \\
  \hline
 \left[\begin{matrix} C_x\\ C_y \end{matrix}\right] & I_{n_x+n_y}
    \end{array} \right),
\end{align}
where $B=[B_x,B_y, B_z]$ and $C^T=[ C_x^T,  C_y^T, C_z^T]$.

Moreover, the ME process obtained by regressing the past values of $ \by^{\rm JEnt} ||\bz$ from $\bx^{\rm JEnt} || \bz$ is equal to
\begin{align}
\bx^{\rm JEnt} ||\by\bz &= \bx^{\rm JEnt} || \bz - G_{xx}(L)^{-1} G_{xy}(L)  \by^{\rm JEnt} ||\bz,\\
&=G_{xx}(L)^{-1} \beps_{x}.\label{eq:JEnt123}
\end{align}
Thus the PSD function of $\bx^{\rm JEnt} ||\by\bz$ is given by 
\begin{align}
S_{\bx^{\rm JEnt} ||\by\bz}(\theta) &= G_{xx}^{-1}(e^{i\theta}) \Omega_{xx}G^{-*}_{xx}(e^{i\theta}).
\end{align}
On the other hand, the PSD function of $\bx^{\rm JEnt} || \bz$ is equal to
\begin{align}\label{eq:Sxz_a}
S_{\bx^{\rm JEnt} ||\bz}(\theta) &= [I_{n_x},\, 0] \left[  \begin{matrix}G_{xx}(e^{i\theta})& G_{xy}(e^{i\theta})\\  G_{yx}(e^{i\theta})& G_{yy}(e^{i\theta})\end{matrix} \right]^{-1} \left[  \begin{matrix}\Omega_{xx}& \Omega_{xy}\\  \Omega_{yx}& \Omega_{yy}\end{matrix} \right] \nonumber\\
&\times \left[  \begin{matrix}G_{xx}(e^{i\theta})& G_{xy}(e^{i\theta})\\  G_{yx}(e^{i\theta})& G_{yy}(e^{i\theta})\end{matrix} \right]^{-*} \left[\begin{matrix}I_{n_x}\\ 0\end{matrix} \right].
\end{align}

By using the SSR in Eq. \eqref{eq:SS_xy_JEnt} and the spectral factorization algorithm, $S_{\bx^{\rm JEnt} ||\by\bz}(\theta)$ can be factorized as
\begin{align}\label{eq:Sxz_b}
S_{\bx^{\rm JEnt}||\bz}(\theta) &=  P_{ x^{\rm JEnt}||z} (e^{i\theta}) \Sigma_{ x^{\rm JEnt}||z}   P_{ x^{\rm JEnt}||z} (e^{i\theta}) ^*,
\end{align}
where $P_{ x^{\rm JEnt}||z} (e^{i\theta})$ is the unique minimum-phase spectral factor. 
Then, cGCM-JEnt and fcGCM-JEnt are given by
\begin{align}
\F_{\by\rightarrow \bx || \bz}^{\rm JEnt} &= 2(h(\bx^{\rm JEnt} ||\bz) - h(\bx^{\rm JEnt} ||\by\bz) ),\nonumber\\
& = \ln \frac{\det  \Sigma_{ x^{\rm JEnt}||z}   }{\det \Omega_{xx}},\label{eq:FIEnt}\\
\f_{\by\rightarrow \bx || \bz}^{\rm JEnt}(\theta) & = \ln \frac{\det S_{\bx^{\rm JEnt} || \bz}(\theta)}{\det S_{ \bx^{\rm JEnt} ||\by\bz}(\theta)},
\end{align}
respectively. 

{
It is noted that $\f_{\by \rightarrow \bx|| \bz}^{\rm SEnt}$ has the highest computational complexity among the three proposed cGCM methods since it involves a spectral factorization of the augmented system in \eqref{eq:Gaug}. 
Both the ME-based fcGCM $\f^{\rm Std-Ent}_{\by \rightarrow \bx | \bz}$ and the original $\f^{\rm Std-Geweke}_{\by \rightarrow \bx | \bz}$ methods have similar computational complexity since they all involve calculation of the representation for the joint process $(\bx; \bz)$ by using the spectral factorization algorithm. 
The joint ME method $\f_{\by\rightarrow \bx || \bz}^{\rm JEnt}$ has the least computational complexity since all the variables are directly provided by the VAR model in \eqref{eq:G_abc}. Though an accurate solution for the corresponding time-domain measure $\F_{\by\rightarrow \bx || \bz}^{\rm JEnt}$ still requires a spectral factorization alogirhtm to obtain $\det  \Sigma_{ x^{\rm JEnt}||z}$, an approximate value for $\det  \Sigma_{ x^{\rm JEnt}||z}$ can be obtained by a discrete approximation for the geometric mean of $\det S_{\bx^{\rm JEnt} || \bz}(\theta)$ based on \eqref{eq:OmegaPSD}.}

%
%

\balance

\bibliographystyle{IEEEtran}

\bibliography{GCM}

\end{document}


\title{Minimum-entropy causal inference and its application in brain network analysis: supplementary materials}
\date{}

\author{Lipeng Ning
\thanks{L. Ning is with the Psychiatry Neuroimaging Laboratory, Department
of Psychiatry, Brigham and Women's Hospital, Harvard Medical School, Boston, MA 02215 USA,
e-mail: lning@bwh.harvard.edu.}}

\markboth{Journal of \LaTeX\ Class Files,~Vol.~xx, No.~xx, April~2021}%
{L. Ning: ME causal inference and its application in brain network analysis}

\maketitle

\makeatletter 
\renewcommand{\thefigure}{S\@arabic\c@figure}
\makeatother

\section*{Additional simulation results}
This section provides additional results for the simulation experiments. Fig. \ref{fig:sim1} is a different version of Fig. 2 of the manuscript where the plots in the last two rows have a different axis scale. The receiver operating characteristic (ROC) plots in the third row of Fig. \ref{fig:sim1} show that the frequency-domain cGCM-SEnt provides higher accuracy for detecting network connections compared to other methods. For the first two network structures, the frequency-domain cGCM-SEnt has higher accuracy than the corresponding time-domain measures. All time-domain methods have high accuracy for relatively simple start-shaped networks.
Fig. \ref{fig:eig2_FPR} provides more details about the trade-off between the false positive rate (FPR) and the true positive rate (TPR) versus the significance level $\alpha$ used as a threshold for p-values without corrections for multiple comparisons. The first two rows of Fig. \ref{fig:eig2_FPR} show the results based on the average fGCM or fcGCM for $\theta\in [0\, \tfrac\pi 2]$. The last two rows show the results based on the time-domain GCM or cGCM. The black dashed lines indicate a significance level of $\alpha=0.5$. It is interesting to see that the plots for FPR of the three structures have similar values. But the corresponding TPR plots of the three network structures are very different. 

Table \ref{table:1A} provides more details about the FPR and TPR for different fGCM and fcGCM measures for the three structures. 
For each measure, three methods were used in statistical testing which includes the Bonferroni \cite{Dunn1961} method for multiple comparisons, the Benjamini-Hochberg procedure \cite{Benjamini1995} for controlling false discovery rate (FDR) and the standard method without corrections for multiple comparisons with $\alpha = 0.05$. The threshold used to reject the null hypothesis gradually increases between the three methods. As a result, the Bonferroni method has the lowest FPR and TPR. For all three structures, cGCM-SEnt and cGCM-JEnt with Bonferroni correction has the highest TPRs with the FPR for cGCM-SEnt slightly lower than 0.05 and the FPR for cGCM-JEnt slightly higher than 0.05. 

Table \ref{table:1B} shows more detailed comparisons of the FPR and TPR based on time-domain GCM and cGCM. The performance of the Bonferroni method is more conservative compared to the results in Table \ref{table:1A} since the FPRs are much lower than $0.05$ for all three structures for the three cGCM methods. The FPR and TPR for the cGCM-SEnt and cGCM-JEnt are all lower than the corresponding results shown in Table \ref{table:1A}.

\begin{table}[h!]
\centering
\resizebox{\textwidth}{!}{
\begin{tabular}{|c |c| c| c| c| c| c|c|c|c|c|c|c|} 
 \hline
 \multicolumn{13}{|c|}{Structure (a)}\\
 \hline
 &\multicolumn{4}{| c |}{frequency-domain GCM}  &\multicolumn{8}{| c |}{frequency-domain conditional GCM} \\ 
 \hline
 & \multicolumn{2}{| c |}{Geweke} &  \multicolumn{2}{| c |}{Ent}&  \multicolumn{2}{| c |}{Geweke} &  \multicolumn{2}{| c |}{Std-Ent} & \multicolumn{2}{| c |}{SEnt} &  \multicolumn{2}{| c |}{JEnt}  \\ \hline
 & FPR & TPR  & FPR & TPR  & FPR & TPR  & FPR & TPR  & FPR & TPR  & FPR & TPR \\ \hline
 Bonferroni & 0.1337 & 0.7716& 0.4128& 0.9798& 0.0016& 0.3537& 0.0018& 0.4251& 0.0264& 0.8165& 0.0699& 0.9126\\ \hline
Benjamini-Hochberg &0.4294& 0.9596&	0.7266&	0.9988&	0.0072&	0.5415&	0.0096&	0.6275&	0.1088&	0.9540&	0.1841&	0.9796\\ \hline
No correction &0.5282&0.9819&	0.7628&	0.9995&	0.0555&	0.8521&	0.0658&	0.8826&	0.2156&	0.9858&	0.2727&	0.9906 \\ \hline
 \multicolumn{13}{|c|}{Structure (b)} \\ \hline
 Bonferroni &     0.1404&    0.5457&    0.5329&    0.9485&    0.0018&    0.1701&    0.0029&    0.2019&    0.0334&    0.6298&    0.0842&    0.7724\\ \hline
Benjamini-Hochberg &0.4532&    0.8807&    0.8425&    0.9969&    0.0065&    0.2831&    0.0108&    0.3451&    0.1401&    0.8762&    0.2188&    0.9263\\ \hline
No correction &0.5470&    0.9337&    0.8683&    0.9982&    0.0606&    0.6636&    0.0771&    0.6957&    0.2403 &   0.9408 &   0.2966&    0.9586 \\ \hline
 \multicolumn{13}{|c|}{Structure (c)} \\ \hline
 Bonferroni &    0.1544&    0.9928&    0.6787&    0.9996&    0.0200&    0.7652&    0.0131&    0.3217 &   0.0267 &   0.8636 &   0.0897 &   0.8501\\ \hline
Benjamini-Hochberg &0.4534&    0.9998&    0.8644&    1.0000 &   0.0656&    0.8975&    0.0328 &   0.4984&    0.0827 &   0.9685  &  0.1911&    0.9014\\ \hline
No correction &0.5589&    1.0000&    0.8839&    1.0000&    0.1850&    0.9949&    0.1437 &   0.8549 &   0.2089 &   0.9990 &   0.2922&    0.9264 \\ \hline
 \end{tabular}}
\caption{Comparison of FPR and TPR related to different average fGCM or fcGCM measures for $\theta\in[0\, \tfrac\pi 2]$ for the three structures shown in the first row of Fig. \ref{fig:sim1} with $\lambda_{\rm max}=0.85$ for the VAR model of each structure.}
\label{table:1A}
\end{table}

\begin{table}[h!]
\centering
\resizebox{\textwidth}{!}{
\begin{tabular}{|c |c| c| c| c| c| c|c|c|} 
 \hline
 \multicolumn{9}{|c|}{Structure (a)}\\
 \hline
 &\multicolumn{2}{| c |}{\multirow{2}{*}{GCM}}  & \multicolumn{6}{| c |}{conditional GCM} \\ 
 \cline{1-1}\cline{4-9}
 & \multicolumn{2}{| c |}{}  &  \multicolumn{2}{| c |}{Std} & \multicolumn{2}{| c |}{SEnt} &  \multicolumn{2}{| c |}{JEnt}  \\ \hline
 & FPR & TPR  & FPR & TPR  & FPR & TPR  & FPR & TPR   \\ \hline
 Bonferroni & 0.0576 &	0.7106&	0.0004&	0.2235&	0.0001&	0.0926&	0.0024&	0.3733\\ \hline
Benjamini-Hochberg &0.2543&	0.9300&	0.0018&	0.3474&	0.0003&	0.1124&	0.0102&	0.5653\\ \hline
No correction &0.3931&	0.9755&	0.0296&	0.7836&	0.0109&	0.6651&	0.0666&	0.8546 \\ \hline
 \multicolumn{9}{|c|}{Structure (b)} \\ \hline
 Bonferroni &     0.0597&    0.4974&    0.0007 &   0.1022&    0.0001&    0.0396&    0.0033&    0.2044\\ \hline
Benjamini-Hochberg &0.3077&    0.8532&    0.0019&    0.1450&    0.0002&    0.0292&    0.0123&    0.3423\\ \hline
No correction &0.4375&    0.9312 &   0.0366&    0.5859 &   0.0144 &   0.4490  &  0.0778&    0.6882 \\ \hline
 \multicolumn{9}{|c|}{Structure (c)} \\ \hline
 Bonferroni &     0.0907&    0.9765 &   0.0023 &   0.7050&    0.0001&    0.4178&    0.0156&    0.8648\\ \hline
Benjamini-Hochberg &0.3354&    0.9975&    0.0082&    0.7171&    0.0005 &   0.5178 &   0.0519 &   0.8744\\ \hline
No correction &0.4741 &   0.9995&    0.0714&    0.9905&    0.0222&    0.9659 &   0.1615 &   0.9962 \\ \hline
 \end{tabular}}
\caption{Comparison of FPR and TPR related to different time-domain GCM or cGCM measures for the three structures shown in the first row of Fig. \ref{fig:sim1} with $\lambda_{\rm max}=0.85$ for the VAR model of each structure}
\label{table:1B}
\end{table}

Fig. \ref{fig:sim2} shows the simulation experiments corresponding to system matrices with $\lambda_{\rm max} =0.6$. 
The second row shows the fGCM and fcGCM functions of the non-zero connections of the simulated network structures.
The third row shows that the frequency-domain cGCM-SEnt and cGCM-JEnt measures have higher accuracy than other measures. The relatively weaker system dynamic reduces the accuracy for all methods compared to results in Fig. \ref{fig:sim1}. The frequency-domain cGCM-SEnt and cGCM-JEnt measures still have higher accuracy than the time-domain methods shown in the last row of Fig. \ref{fig:sim2}. With a relatively high noise level, the cGCM measures have slightly worse performance than the GCM measures for the first two network structures. Fig. \ref{fig:eig1_FPR} shows the corresponding FPR and TPR plots versus the significant levels. All methods have lower FPR and TPR compared to results in Fig. \ref{fig:eig2_FPR}. Similar to results in Fig. \ref{fig:eig2_FPR}, all three structures also have similar FPR plots. The time-domain cGCM-SEnt methods have much lower FPR and TPR compared to the corresponding frequency-domain measures. 

Table \ref{table:2A} and Table \ref{table:2B} illustrate the FPR and TPR with different methods for multiple comparisons with $\lambda_{\rm max} =0.6$ in the underlying systems. The FPR and TPR are all lower than the corresponding results in Table \ref{table:1A} and Table \ref{table:1B}. fcGCM-SEnt and fcGCM-JEnt with the Bonferroni method still provide better performance than other methods.

\begin{table}[h!]
\centering
\resizebox{\textwidth}{!}{
\begin{tabular}{|c |c| c| c| c| c| c|c|c|c|c|c|c|} 
\hline
 \multicolumn{13}{|c|}{Structure A}\\
 \hline
 &\multicolumn{4}{| c |}{frequency-domain GCM}  &\multicolumn{8}{| c |}{frequency-domain conditional GCM} \\ 
 \hline
 & \multicolumn{2}{| c |}{Geweke} &  \multicolumn{2}{| c |}{Ent}&  \multicolumn{2}{| c |}{Geweke} &  \multicolumn{2}{| c |}{Std-Ent} & \multicolumn{2}{| c |}{SEnt} &  \multicolumn{2}{| c |}{JEnt}  \\ \hline
 & FPR & TPR  & FPR & TPR  & FPR & TPR  & FPR & TPR  & FPR & TPR  & FPR & TPR \\ \hline
 Bonferroni & 0.0420&0.3406&	0.2976&	0.8516&	0.0018&	0.1064&	0.0030&	0.1584&	0.0352&	0.5237&	0.0818&	0.6858\\ \hline
Benjamini-Hochberg &0.1598&	0.6176&	0.5965&	0.9663&	0.0041&	0.1116&	0.0075&	0.2211&	0.1269&	0.7855&	0.2020&	0.8635\\ \hline
No correction &0.2945&	0.8156&	0.6493&	0.9766&	0.0628&	0.5349&	0.0777&	0.6386&	0.2368&	0.8927&	0.2902&	0.9181\\ \hline
 \multicolumn{13}{|c|}{Structure (b)} \\ \hline
 Bonferroni &    0.0410&    0.1931&    0.3452&    0.7718&    0.0017&    0.0520 &   0.0029&    0.0812&    0.0373&    0.3942 &   0.0880&    0.5504\\ \hline
Benjamini-Hochberg &0.1608&    0.4484&    0.6807&    0.9529&    0.0036&    0.0500 &   0.0068 &   0.1038  &  0.1414&    0.6773 &   0.2188&    0.7768\\ \hline
No correction &0.3051&    0.6966&    0.7244&    0.9657&    0.0636&    0.3882 &   0.0809&    0.4823&    0.2474&    0.8102&    0.3001&    0.8507 \\ \hline
 \multicolumn{13}{|c|}{Structure (c)} \\ \hline
 Bonferroni &    0.0530&    0.7755&    0.3329&    0.9700&    0.0097&    0.3256&    0.0068&    0.2449 &   0.0192 &   0.6488&    0.0648&    0.8734\\ \hline
Benjamini-Hochberg &0.1815&    0.9181&    0.6704&    0.9966&    0.0293&    0.5014 &   0.0139 &   0.3270  &  0.0663 &   0.8582  &  0.1499&    0.9516\\ \hline
No correction &0.3331&    0.9800&    0.7177&    0.9974&    0.1351&    0.8741 &   0.1071&    0.7751 &   0.1885 &   0.9775 &   0.2541 &   0.9770 \\ \hline
 \end{tabular}}
\caption{Comparison of FPR and TPR related to different average fGCM or fcGCM measures for $\theta\in[0\, \tfrac\pi 2]$ for the three structures shown in the first row of Fig. \ref{fig:sim1} with $\lambda_{\rm max}=0.6$ for the VAR model of each structure.}
\label{table:2A}
\end{table}

\begin{table}[h!]
\centering
\resizebox{\textwidth}{!}{
\begin{tabular}{|c |c| c| c| c| c| c|c|c|} 
 \hline
 \multicolumn{9}{|c|}{Structure (a)}\\
 \hline
 &\multicolumn{2}{| c |}{\multirow{2}{*}{GCM}}  & \multicolumn{6}{| c |}{conditional GCM} \\ 
 \cline{1-1}\cline{4-9}
 & \multicolumn{2}{| c |}{}  &  \multicolumn{2}{| c |}{Std} & \multicolumn{2}{| c |}{SEnt} &  \multicolumn{2}{| c |}{JEnt}  \\ \hline
 & FPR & TPR  & FPR & TPR  & FPR & TPR  & FPR & TPR   \\ \hline
 Bonferroni & 0.0206& 	0.3036&	0.0009&	0.0666&	0.0001&	0.0206&	0.0038&	0.1373\\ \hline
Benjamini-Hochberg &0.0837&	0.5266&	0.0014&	0.0508&	0.0001&	0.0077&	0.0087&	0.1689\\ \hline
No correction &0.2153&	0.7916&	0.0409&	0.4616&	0.0178&	0.3346&	0.0807&	0.5681 \\ \hline
 \multicolumn{9}{|c|}{Structure (b)} \\ \hline
 Bonferroni &     0.0187&    0.1681&    0.0008&    0.0329&    0.0001&    0.0096 &   0.0041&    0.0785\\ \hline
Benjamini-Hochberg &0.0857&    0.3698 &   0.0014&    0.0231&    0.0001&    0.0036&    0.0096 &   0.0917\\ \hline
No correction &0.2323 &   0.6799&    0.0449 &   0.3304 &   0.0202&    0.2224 &   0.0871&    0.4347 \\ \hline
 \multicolumn{9}{|c|}{Structure (c)} \\ \hline
 Bonferroni &     0.0223 &   0.6686 &    0.0015  &  0.2994  &  0.0003 &   0.1303 &   0.0068 &   0.4954\\ \hline
Benjamini-Hochberg &0.0968 &   0.8373 &   0.0039 &   0.3914  &  0.0007&    0.1645 &   0.0216 &   0.6118\\ \hline
No correction &0.2504 &   0.9716 &   0.0575 &   0.8602 &   0.0235  &  0.7455  &  0.1117&    0.9214 \\ \hline
 \end{tabular}}
\caption{Comparison of FPR and TPR related to different time-domain GCM or cGCM measures for the three structures shown in the first row of Fig. \ref{fig:sim1} with $\lambda_{\rm max}=0.6$ for the VAR model of each structure}
\label{table:2B}
\end{table}

\begin{figure}[htb]
\centering
\includegraphics[width=.8\textwidth]{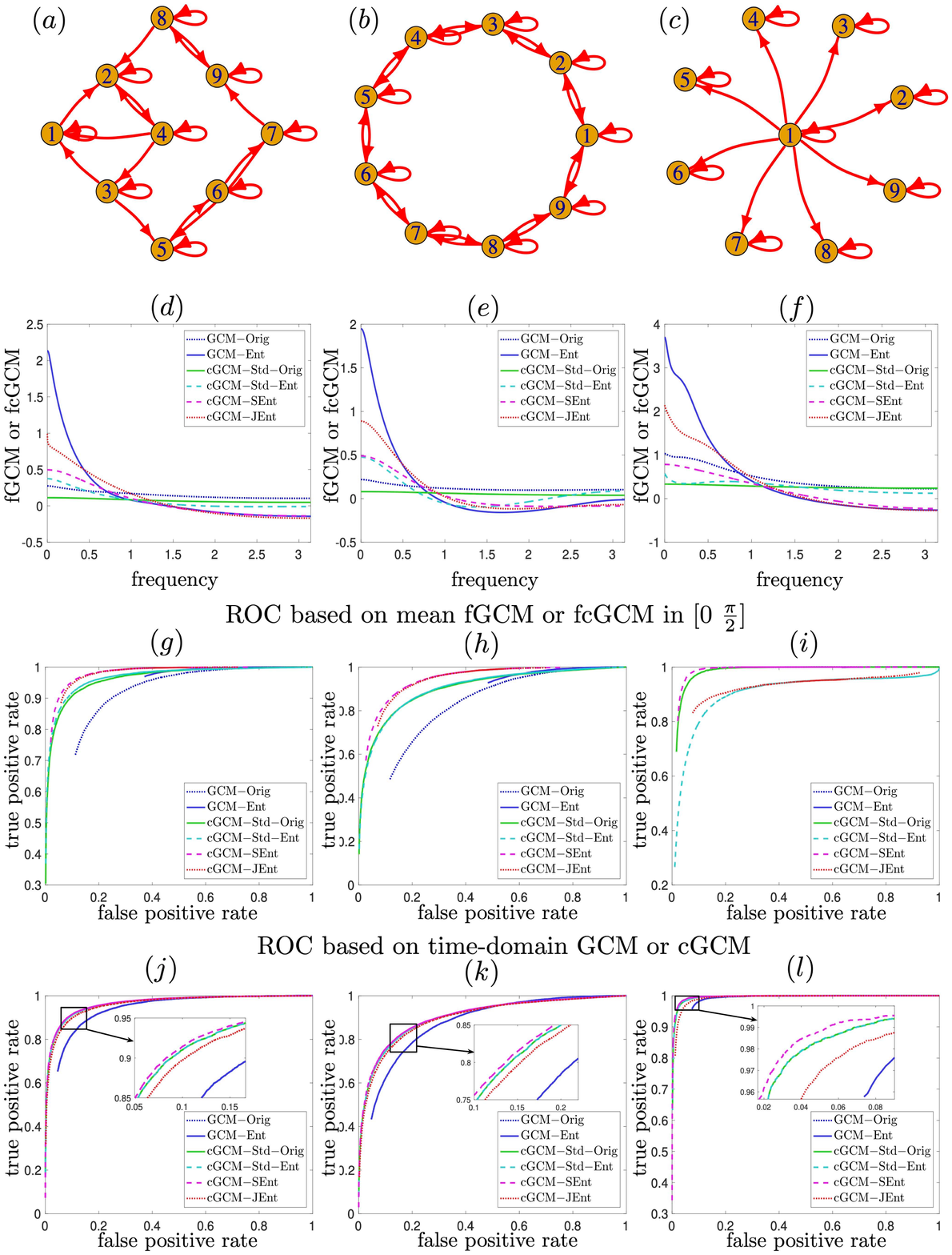}
\caption{Illustration of simulation results corresponding to VAR models with $\lambda_{\rm max} =0.85$. The first row demonstrates the structure of three VAR models used in the simulations. The second row shows the sample mean of fGCM and fcGCM functions for all non-zero connections in the first row.
The third row shows the ROC curves based on mean fGCM and fcGCM values in the frequency interval $[0~\tfrac\pi 2]$.The last row illustrates the ROC curve based on time-domain GCM and cGCM measures.
}\label{fig:sim1}
\end{figure}

\begin{figure}[htb]
\centering
\includegraphics[width=.8\textwidth]{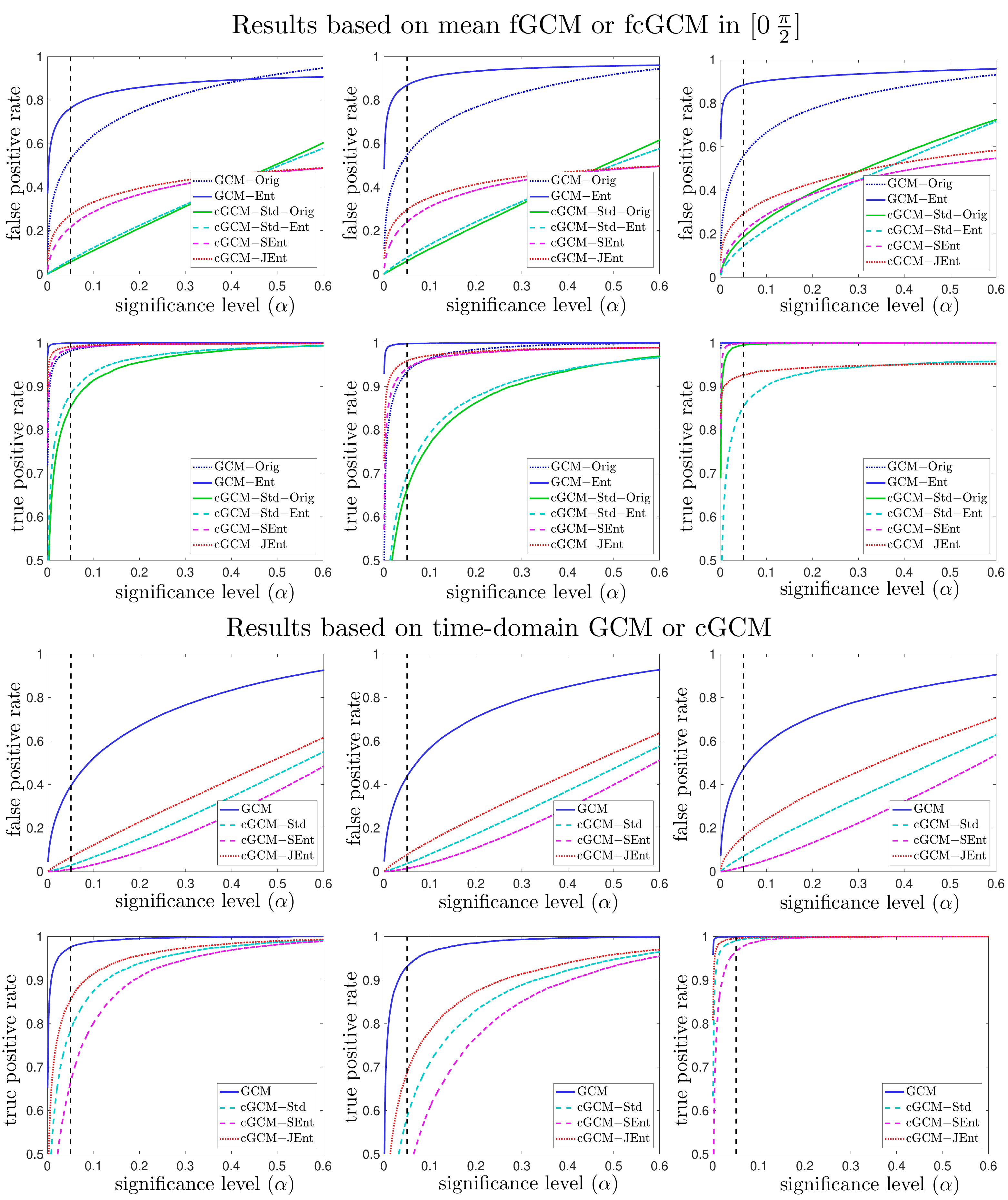}
\caption{Illustration of the FPR and TPR versus the significance level used as the threshold for p-values corresponding to VAR models with $\lambda_{\rm max} =0.85$. Each column corresponds to the three network structures illustrated in the first row of Fig. \ref{fig:sim1} respectively.
The first two rows illustrate results based on the mean fGCM and fcGCM in $[0~\tfrac\pi 2]$. The last two rows show the results corresponding to the time-domain GCM and cGCM.
}\label{fig:eig2_FPR}
\end{figure}

\begin{figure}[htb]
\centering
\includegraphics[width=.8\textwidth]{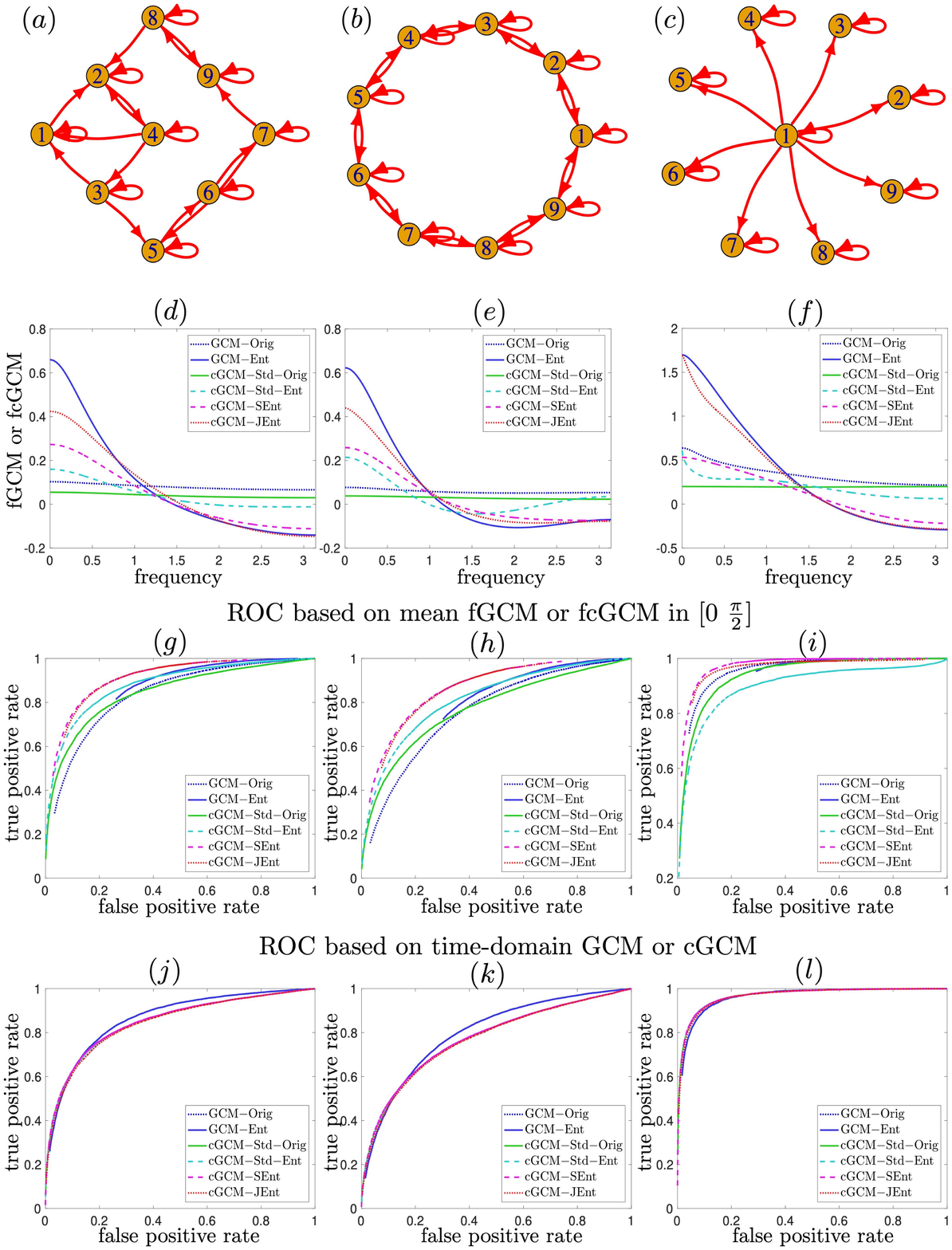}
\caption{Illustration of simulation results corresponding to VAR models with $\lambda_{\rm max} =0.6$. The first row demonstrates the structure of three VAR models used in the simulations. The second row shows the sample mean of fGCM and fcGCM functions for all non-zero connections in the first row.
The third row shows the ROC curves based on mean fGCM and fcGCM values in the frequency interval $[0~\tfrac\pi 2]$. The last row illustrates the ROC curve based on time-domain GCM and cGCM measures.
}\label{fig:sim2}
\end{figure}

\begin{figure}[htb]
\centering
\includegraphics[width=.8\textwidth]{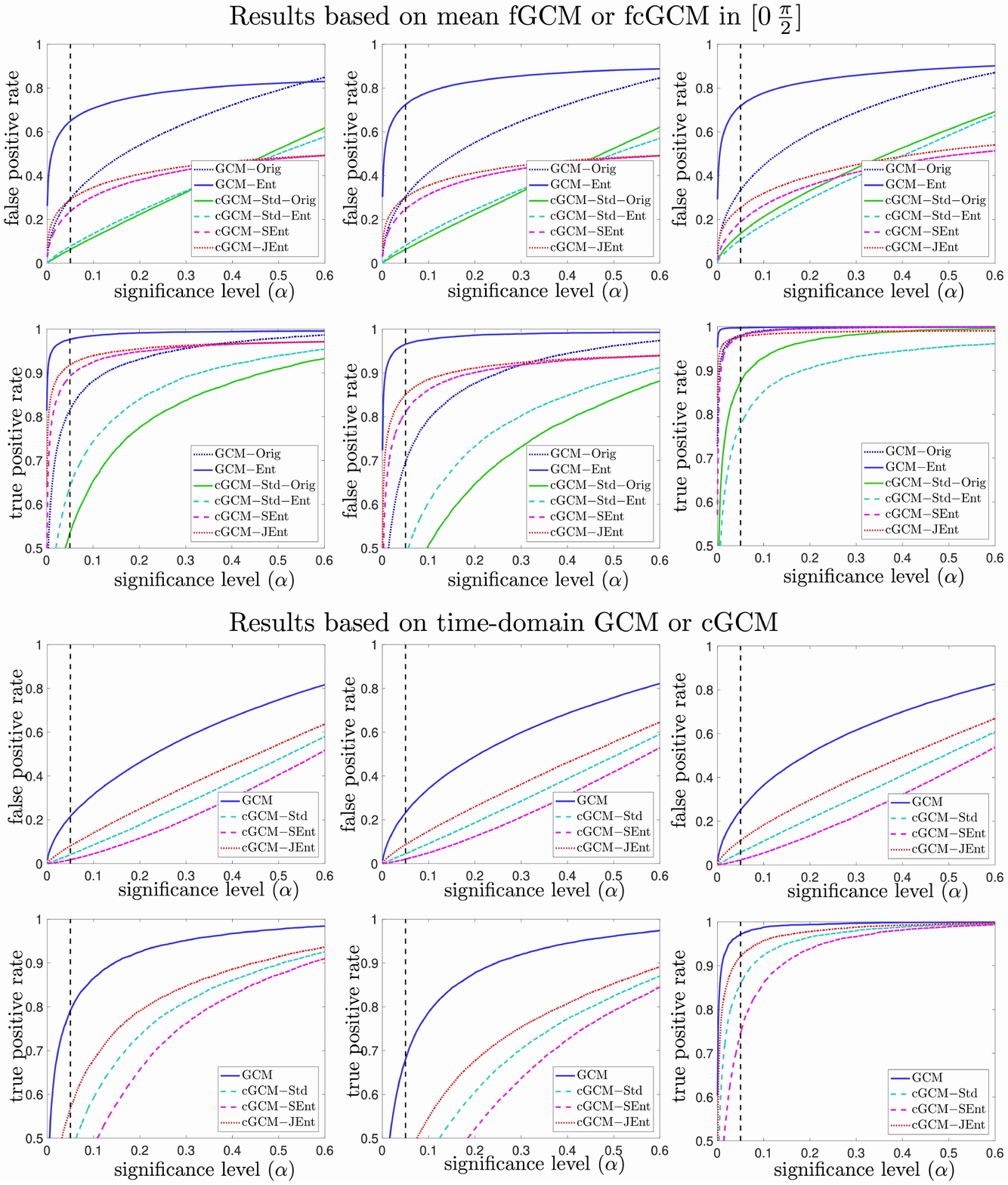}
\caption{Illustration of the FPR and TPR versus the significance level used as the threshold for p-values corresponding to VAR models with $\lambda_{\rm max} =0.6$. Each column corresponds to the three network structures illustrated in the first row of Fig. \ref{fig:sim1} respectively.
The first two rows illustrate results based on the mean fGCM and fcGCM in $[0~\tfrac\pi 2]$. The last two rows show the results corresponding to the time-domain GCM and cGCM.
}\label{fig:eig1_FPR}
\end{figure}

\section*{Additional experimental results on in vivo rsfMRI analysis}

The first row of Fig. \ref{fig:invivo} briefly summarizes the key experimental procedures which include the extraction of rsfMRI time series from each region of interest (ROI) based on the AAL atlas and the estimation of structural connectivity matrix based on diffusion MRI tractography.

Fig. \ref{fig:invivo} illustrates the average SC matrix from 100 subjects between 120 ROIs, which are separated to three groups: left-hemisphere (LH), right-hemisphere (RH), and cerebellum as shown in the horizontal and vertical axes.
Figs. \ref{fig:invivo}(f) and (g) illustrate the average fGCM-Ent and fcGCM-SEnt between each pair of brain regions in the 0.01 to 0.1 Hz frequency range related to the passband of hemodynamic response \cite{Biswal1995,VanDijk2010}. 
The connectivity matrices corresponding to other GCM and cGCM measures are provided in the Supplementary Material.
Figs. \ref{fig:invivo}(h) and (j) show the histogram of the SC, fGCM-Ent, and fcGCM-SEnt connections between different brain regions, including connections within the LH (LH-LH) and the RH (RH-RH), connections between the different regions in the bilateral brain (LH-RH-asym),  symmetric bilateral brain connections (LH-RH-sym), cerebro-cerebellar connections as well as cerebellar connections.
It can be seen that the top 80-100\% strongest connections in fGCM-Ent and fcGCM-SEnt in Figs. \ref{fig:invivo}(i) and (j) contain more bilateral connections than SC in Fig. \ref{fig:invivo}(h).

Figs. \ref{fig:invivo}(k) to (m) demonstrate the top 400 strongest SC, fGCM-Ent and fcGCM-SEnt connections where the thickness of the lines reflects the strength of the connectivity. 
The sum of the weight of all connections from each brain region gives rise to the node weight which is shown in black bins.
In Figs. \ref{fig:invivo}(l) to (m), the colors of the connections are consistent with the color of the brain region that has a stronger causal influence on the other brain region. 
Fig. \ref{fig:invivo}(l) shows that fGCM-Ent has more strong asymmetric bilateral connections whereas Fig. \ref{fig:invivo}(m) shows that fcGCM-SEnt has more strong symmetric bilateral connections.
The significant difference between Fig. \ref{fig:invivo}(k) and  Figs. \ref{fig:invivo}(l), (m) indicates the fiber density from tractography does not fully reflect the strength of effectivity connectivity, particularly for cross-hemispheric connections.
Similar results have been shown in \cite{Deco2014} that adding more cross-hemispheric connections to the SC matrix estimated using dMRI can improve the consistency with functional connectivity.
\begin{figure}[tbhp]
\centering
\includegraphics[width=1\linewidth]{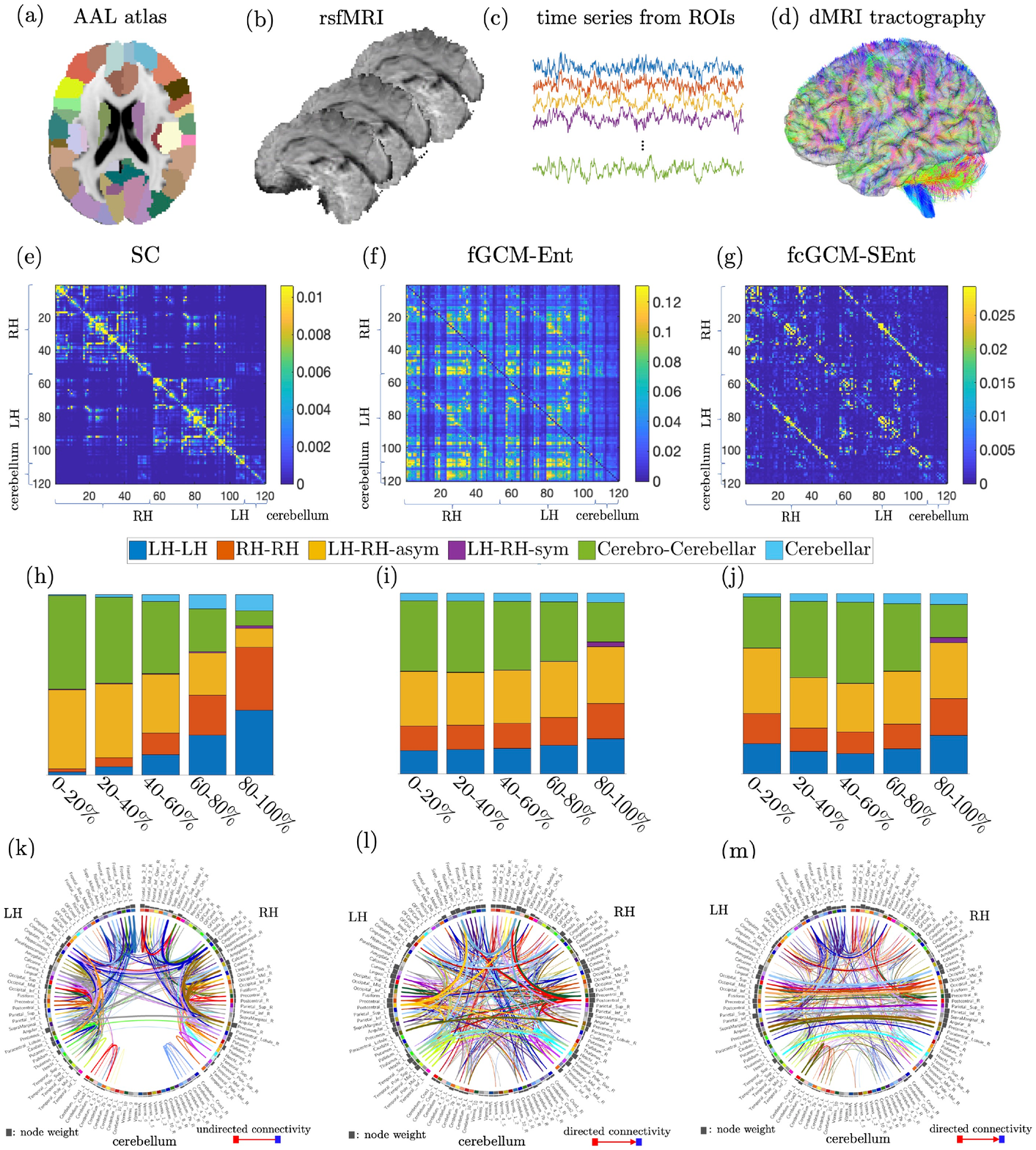}
\caption{Results of in vivo MRI data analysis. (a) illustrates the AAL atlas. (b) shows the rsfMRI volumes. (c) illustrates the rsfMRI time series. (d) illustrates the whole-brain tractography estimated using dMRI. (e) shows the SC matrix between 120 brain regions. (f) and (g) illustrate the mean fGCM-Ent and fcGCM-SEnt between 0.01 and 0.1 Hz. 
(h), (i), (j) show the histogram of SC, fGCM-Ent, and fcGCM-SEnt values for connections between different brain regions.
(k), (l), (m) illustrate the top 400 strongest SC, fGCM-Ent and fcGCM-SEnt connections. 
}
\label{fig:invivo}
\end{figure}

\bibliography{GCM}